\newif\ifdraft
\newif\iflipi
\draftfalse
\lipitrue
\newif\ifbioblabla
\bioblablafalse
\newif\ifarxiv
\arxivfalse

\documentclass[english,cleveref, autoref]{article}
\usepackage{amsmath,amssymb,amsthm}
\usepackage{xspace}
\usepackage{xfrac}
\usepackage{authblk}
\usepackage{subcaption}
\usepackage[margin=1in]{geometry}

\date{}
\title{Noidy Conmunixatipn: On the Convergence of the Averaging Population Protocol} 

\author[1]{Frederik Mallmann-Trenn}
\author[2]{Yannic Maus}
\author[1]{Dominik Pajak}
\affil[1]{MIT, CSAIL, US }
\affil[1]{Department of Computer Science, Technion, Haifa, Israel}

\usepackage{xcolor}
\definecolor{darkgreen}{rgb}{0,0.5,0}
\usepackage{hyperref}
\hypersetup{colorlinks=true, linkcolor=blue, urlcolor=blue, citecolor=darkgreen}
\usepackage{mathtools}
\DeclarePairedDelimiter{\ceil}{\lceil}{\rceil}
\DeclarePairedDelimiter{\floor}{\lfloor}{\rfloor}
\renewcommand{\floor}[1]{\left\lfloor #1 \right\rfloor}
\renewcommand{\ceil}[1]{\left\lceil #1 \right\rceil}

\usepackage{aliascnt}
\def\NewTheorem#1#2{%
  \newaliascnt{#1}{theorem}
  \newtheorem{#1}[#1]{#2}
  \aliascntresetthe{#1}
  \expandafter\def\csname #1autorefname\endcsname{#2}
}

 \newtheorem{theorem}{Theorem}[section]
\NewTheorem{lemma}{Lemma}
\NewTheorem{corollary}{Corollary}
\NewTheorem{definition}{Definition}
\NewTheorem{remark}{Remark}

\NewTheorem{fact}{Fact}

\NewTheorem{conjecture}{Conjecture}
\NewTheorem{observation}{Observation}
\NewTheorem{assumption}{Assumption} 
\NewTheorem{assumptions}{Assumptions}

\NewTheorem{proposition}{Proposition}

\NewTheorem{claim}{Claim}

\usepackage[inline]{enumitem}
\setenumerate{label=(\roman*),itemsep=0pt,topsep=0pt,partopsep=0pt}

\newcommand{\eps}{\varepsilon}

\newcommand{\norm}[1]{\left\lVert#1\right\rVert}

\newcommand{\Val}[2]{\ensuremath{X_{#1}^{(#2)}}}
\newcommand{\val}[2]{\ensuremath{x_{#1}^{(#2)}}}
\newcommand{\Valvec}[1]{\ensuremath{\mathbf{X}^{(#1)}}}
\newcommand{\valvec}[1]{\ensuremath{\mathbf{x}^{(#1)}}}

\newcommand{\avg}[1]{\ensuremath{\varnothing^{(#1)}}}

\newcommand{\Gaussian}{Gaussian white noise model\xspace}
\newcommand{\Disc}{Discrete white noise model\xspace}



\renewcommand{\Pr}[1]{\mathbb{P}\left[\,#1\,\right]}
\newcommand{\Var}[1]{\operatorname{Var}\left[\,#1\,\right]}
\newcommand\E[1]{\mathbb{E}\left[\,#1\,\right]}

\usepackage[textsize=tiny]{todonotes}

\newcommand{\makenote}[2]{{{\color{#1} #2}}}

\newcommand{\fnote}[1]{\ifdraft \makenote{blue}{FM: #1} \fi}

\newcommand{\arxiv}[1]{\ifdraft \ifarxiv \makenote{purple}{Arxiv: #1} \fi\fi}

\newcommand{\ynote}[1]{\ifdraft \todo[color=green!80]{YM: #1} \fi}

\newcommand{\dnote}[1]{\ifdraft\makenote{darkgreen}{D: #1}\fi}

\renewcommand{\paragraph}[1]{ \medskip{\bf #1}.\xspace}

\newcommand{\ie}{\textit{ i.e.,}\xspace}
\newcommand{\eg}{\textit{e.g.,}\xspace}

\begin{document}
\maketitle

\begin{abstract}
We study a process of \emph{averaging} in a distributed system with \emph{noisy communication}. Each of the agents in the system starts with some value and the goal of each agent is to compute the average of all the initial values. In each  round, one pair of agents is drawn uniformly at random from the whole population, communicates with each other and each of these two agents updates their local value based on their own value and the received message. The communication is noisy and whenever an agent sends any value $v$, the receiving agent receives $v+N$, where $N$ is a zero-mean Gaussian random variable.
The two quality measures of interest are (i) the total sum of squares $TSS(t)$, which measures the sum of square distances from the 
average load to the \emph{initial average} and (ii) $\bar{\phi}(t)$, measures the sum of square distances from the 
average load to the \emph{running average} (average at time $t$).

 It is known that the simple averaging protocol---in which an agent sends its current value and sets its new value to the average of the received value and its current value---converges eventually to a state where $\bar{\phi}(t)$ is small.
 It has been observed that $TSS(t)$, due to the noise, eventually diverges and previous research---mostly in control theory---has focused on showing eventual convergence w.r.t. the running average. 
 We obtain the first probabilistic bounds on the convergence time of $\bar{\phi}(t)$ and precise bounds on the drift of $TSS(t)$ that show
 that albeit $TSS(t)$ eventually diverges, for a wide and interesting range of parameters, $TSS(t)$ stays small for a number of rounds that is polynomial in the number of agents. 
  Our results extend to the synchronous setting and settings where the agents are restricted to discrete values and perform rounding.

\vfill
{\tiny 
\noindent
Frederik Mallmann-Trenn and Dominik Pajak were supported in part by NSF Award Numbers CCF-1461559, CCF-0939370, and CCF-18107. Yannic Maus was partly supported by ERC Grant No. 336495 (ACDC).
}

\end{abstract}
\newpage
\tableofcontents

\section{Introduction}
We consider the problem of distributed averaging by a group of agents (\eg sensors), initialized with values that represent, for example, different temperature measurements. The agents' goal is to compute the average of all the initial values using the following simple dynamic: In each discrete round, two agents are drawn uniformly at random from the whole population, communicate their values to each other and set their new values to the average of their old value and the received value.
%
Converging to the average plays a key-role in many applications, e.g., for sensor networks  \cite{xiao2005scheme,schizas2007consensus}, social insects \cite{brumm2013animal}, and robotics \cite{EGN18,GORN17}. 
In all of these applications, the agents  (sensors,  ants, and robots) are very simple and are therefore limited in both memory and communication. Moreover, communication is often erroneous.\footnote{Consult \autoref{sec:relatedWork} for a more detailed review of these applications including the limitation of agents and further motivation. \autoref{sec:relatedWork} also contains related work on the averaging protocol.}
This motivates the study of the aforementioned simple averaging dynamic in a setting where the agents only remember one value, do not use any additional memory, and the communication is subject to noise. We model the noise in the communication as follows:
 Whenever an agent sends any value $v$, the receiving agent
receives $v+N$, where random variable $N$ is distributed according to some zero-mean probability distribution $\aleph$, \eg a normal distribution. The agents update their values as follows: whenever two agents communicate, each agent sets its new value to the average of their old value and the received value; note that---due to the noise---the two agents might have distinct new values.

The values of the $n$ nodes in step $t$ of the process are denoted by $\Val{1}{t},\Val{2}{t},\dots,\Val{n}{t}$.
We consider the following models:  (i) the \emph{sequential setting} where one pair of agents is chosen uniformly at random and (ii) the \emph{synchronous setting} where each agent is matched to exactly one other agent chosen uniformly at random.  
The   two quality measures of the convergence used in this work are 
 (i) the total sum of squares $TSS(t)=\sum_i (\Val{i}{t} - \avg{0})^2$,  where $\avg{0} = \sum_i \Val{i}{0}/n$ is the initial average and (ii) the sum of squared distances to the running average $\bar{\phi}(t) = \sum_i (\Val{i}{t} - \avg{t})^2$,  where $\avg{t} = \sum_i \Val{i}{t}/n$ is the \emph{running average}.
Our contributions can be informally summarized as follows:

\begin{enumerate}%
\item  We give, under mild assumptions on the noise,
 the first bounds on the convergence time of the running average $\bar{\phi}(t)$  in the noisy gossip-based communication setting. The bounds we obtain are---up to a constant factor---tight. In particular,  the potential converges to a value that is linear in $n$ and the second moment of the noise $\E{N^2}$; which is tight. 
  So far it was only known that the process  eventually converges to a state where $\bar{\phi}(t)$ is small (\eg \cite{xiao2004fast}), but precise bounds were not known. 
 (\autoref{thm:runningavg})
\item

We show that, in contrast to the current belief, one can hope to converge to the \emph{initial} average in addition to convergence to the \emph{running} average as long as the number of rounds are bounded:
 It was known that  $TSS(t)$, due to the noise, eventually diverges (the running average diverges from the initial average) and for this reason related research---mostly in control theory---has focused on showing eventual convergence w.r.t. $\bar{\phi}(t)$; leaving $TSS(t)$ aside.
Since we give precise bounds on the convergence time of the running average, we can show the following. Under mild assumptions on the noise,   $TSS(t)$  converges to almost the same value as $\bar{\phi}(t)$ as long as the number of time steps $t$ is bounded by $O(n^2)$, where $n$ is the number of nodes.  (\autoref{thm:initialavg})


\item  We pioneer in the discrete setting in which the agents can store only integer values and the noise is also an integer. In this setting the agents in our algorithm perform randomized rounding. We show that this only causes a negligible difference from the continuous case. (\autoref{thm:rounding})
\item We study both the sequential and the synchronous setting and show that there is no significant difference (up to a scaling of time) between the models.
(\autoref{thm:synch})

\item We perform simulations in the setting where nodes are limited in storage, \ie they can only store values from a bounded range. This leads to a much faster (by order of magnitude) divergence between the running average and the initial average.
Our simulations also seem to indicate strong bounds on the distribution of distances to the running average in our main model (unbounded values).
(\autoref{sec:simulations})
\end{enumerate}

The convergence time of the averaging  processes  in the gossip-based communication setting \emph{without} noise has been studied before
(e.g., \cite{kempe}). However, to the best of our knowledge, no bounds on the convergence time  are known in the gossip-based communication setting with noise. We continue with a detailed motivation for studying noise in the simple averaging dynamic and related work.

\subsection{Motivation and Related Work}\label{sec:relatedWork}

 Converging to the average  plays a key role in many applications in which agents have limited computational and communication power, e.g.,
\begin{enumerate}[label=(\roman{*})]
\item
 sensor networks \cite{xiao2005scheme,schizas2007consensus}: here there is a wide range of application including terrain monitor applications \cite{simic2003distributed}, computing an average temperature, PIR sensors measuring the  infrared light radiation emitted from objects, and many more applications. In such scenarios links are often faded \cite{rappaport1996wireless,chen2004channel},
\item  social insects: for ants, values could represent the individuals' different assessments of nest qualities when house hunting \cite{brumm2013animal}
 or the deficit of workers at a given task \cite{libbrecht2013interplay}, and 
\item  robotics \cite{EGN18,GORN17} and in particular memory-limited robots, \eg Kilobots exploring the percentage of white tiles in an area \cite{kilobots}, or microbots measuring the concentration of chemicals. 
\end{enumerate}
In all of these applications the agents (representing sensors, ants or robots) are very simple and severely limited in both memory and communication. Moreover, the communication is often not only limited but also erroneous (e.g., consider wireless communication with obstacles between robots), or received messages are subject to interpretation (e.g., when insects communicate through gestures~\cite{leonhardt2016ecology}). Motivated by this unreliable communication in applications we study the simple averaging dynamic where the communication is subject to noise. 

\medskip

We continue with related work. The problem of distributed values converging to the average (often without noise) has been studied in various areas reaching back to early versions studied in statistics~\cite{degroot1974reaching, french1983group,gilardoni1993reaching}.
However, to the best of our knowledge, none of the studied models match our model. We review the related work by areas:
\begin{enumerate*}
\item average consensus and its applications,
\item gossip-based communication models,
\item consensus protocols in population protocols,
\item biological distributed algorithms,
\item noise and failures in sensor networks.
\end{enumerate*}
\fnote{moved this befcause it messed with white space}
\ynote{These two sentences are strongly contradicting each other}
\fnote{why?}
\ynote{Because the first one says no one ever checked convergence time and the next one says ppl checked convergence rate=convergence time}
\fnote{but is says noiseless}

\paragraph{Average consensus and its applications} 
Consensus has been studied intensively in various settings in general network topologies, much of it under the name of \emph{average consensus} \cite{AvgCXiao07,AvgCXiao03}. Most of this work is orthogonal to our work: First, due to the general network topology and the fact that, in each step of the studied algorithms, the agents update their values with a weighted average of \underline{all} of their neighbors' values whereas in our averaging dynamic, an agent can only access a single other value per interaction.  Second, while the potential functions in these works  and the noise, if any, are usually identically or similarly defined as in our work the main goal of these papers is---just as in the classic works---to study under which circumstances the processes eventually converge to a state with a small potential function \cite{AvgCXiao07}, whereas we are interested in the number of interactions until our process obtains a small potential. 
Recent papers~\cite{nedic2017convergence,bu2018accelerated,li2017robust,chen2017critical} consider the convergence rate of the weighted averaging process, but only in the noiseless setting. 
Average consensus has also been studied in networks with time-varying topologies \cite{moreau2005stability,saldana2017resilient}. Variants with noisy communication were studied \cite{AvgCXiao07,kar2009distributed}, but they only consider additive  noise and assume it to be zero-mean with unit variance (as mentioned before, only convergence in the limit is shown). 
The noisy version of the problem also received ample attention in  control theory \cite{touri2009distributed,ren2005survey,ren2008distributed}. 
Already in the early works on average consensus immediate applications of converging to the average were discovered and intensively studied, e.g., applications to load balancing between parallel machines \cite{AvgCLoadB90,AvgCCybenko89} or to coordinate distributed mobile agents \cite{AvgCLoadB90,avgCMobileAGents03,fax2004information}. 
For a more detailed overview on average linear consensus consult the survey \cite{Garin2010}.

\paragraph{Gossip-based communication models} 
Much closer to our work is the study of aggregating information in gossip-based model. In this model, each node can contact \underline{one} of its neighbors in the network in each round and exchange information with it. Even though a node can be contacted by many neighbors in a single round, this model, if applied to the complete graph, is very similar to our synchronous model. On the complete graph \cite{kempe} shows that $O(n\cdot \ln n)$ interactions are enough to approximate the average well with high probability. 
 On the one hand they consider more general graphs (in some sense we consider the complete graph); on the other hand they do not consider noise, which simplifies their analysis of the convergence time significantly.

\paragraph{Consensus protocols in population protocols, biological distributed algorithms} 
Motivated by biological applications, population protocols have also been studied in the noisy setting in the context of  biological distributed algorithms.
The authors of \cite{FHK14} study rumor spreading and consensus in extremely faulty networks where a bit in a message can be flipped with probability $1/2-\eps$. This was later generalized in \cite{FN16} to plurality consensus.
The authors of \cite{BoczkowskiFKN18}  study the differences between pull and push rumor spreading in the noisy setting.
Reaching consensus to an opinion in population protocols in the noiseless setting has received much attention (see \eg
\cite{AngluinADFP06,ElsasserFKMT16,AlistarhAEGR17,AlistarhAG17,BecchettiCNPST17,BerenbrinkCEKMN17,Doty2018,BerenbrinkEFKKR18,KosowskiU18,Przemprzem,GasieniecS18,KanadeMS19}).






\paragraph{Noise and failures in sensor networks} 
The problem of converging to the average (and similar problems) have also been studied in (noisy) sensor networks \cite{xiao2005scheme,schizas2007consensus} where nodes again can interact with all their neighbors. 
In these networks another type of unreliable communication, i.e., packages might be dropped, has received ample attention, e.g., \cite{Censor-HillelHH17} studies  the broadcast problem and \cite{CHHZ18} develops a framework to transform certain algorithms for failure free networks to also work in faulty sensor networks.  

An interesting type of failure has been studied in \cite{GilbertN17}. There failures do not happen during the communication but the algorithm itself might be faulty, i.e., a state machine run at an agent might switch to a wrong state.

\subsection{Formal Results}\label{sec:formalResults}
We now formally state our main theorems. For the ease of presentation, in the discussion we assume that noise is normally distributed with unit variance, $N\sim \mathcal{N}(0,1)$, but our results hold for general variance $\sigma^2$. Let $\phi_0 = \bar{\phi}(\Valvec{0})$ be the initial potential. Our first theorem shows that the agents converge to a small value of $\bar{\phi}(t)=O(n)$ after parallel time\footnote{Recall that in parallel time we scale time by a factor of $n$ for a fair comparison with the synchronous time model.}  that is logarithmic in $\phi_0/n$.
In particular, if we use $b$ to denote the 
initial imbalance ($b=\max_{i,j}\{\val{i}{0}-\val{j}{0}\}$), then 
it takes  $O(\ln b)$ parallel steps for the potential to become $\bar{\phi}(t)=O(n)$. Note that $\bar{\phi}(t)=O(n)$ means that the `average' difference between the values of any two agents is constant and  we show that the constant hidden in the $O$-notation is actually very small.
It is worth mentioning that this is tight in two senses: 
(i) In expectancy we have $\bar{\phi}(t)=\Omega(n)$ for any fixed time step $t\geq n$, (i.e., after one parallel time step). Even in the case where all nodes initially have the same value, our results show that the potential increases after $n$ interactions in expectation by $\Omega(n \E{ N^2})=\Omega(n)$.
 (ii) At least $\Omega(\ln b)$ parallel time steps are required\footnote{For the case where constant fraction of the values are at distance $b$.} to decrease the potential to $O(n)$, since the potential only drops in expectation by a constant factor in each parallel step. The formal statement is as follows.
  \begin{theorem}[Convergence to Running Avg.]\label{thm:runningavg}
		Consider any  noise-distribution  $\aleph$ with (at least) exponential-decay\footnote{In fact we only require the function to be smooth, which we define later. This class is much broader and contains most of the famous distributions including the normal distribution, geometric distribution and the Poisson distribution.    }.
  Fix any $\delta \in \mathbb{R}$. Let $n=n(\delta)$ be large enough. The following hold:
 
	\begin{enumerate}
	\item for any $t=\Omega\left( n \ln\left(\frac{\phi_0}{\delta \sigma^2 n}\right) \right)$ with probability at least $1-\delta$ we have  $\bar{\phi}(\Valvec{t}) = O(\sigma^2 n \ln(1/\delta))$~, 
	\item for any $t\geq n$ (parallel time) with constant probability we have $\bar{\phi}(\Valvec{t}) = \Omega(\sigma^2 n)$~and
\item  even without noise, for any   $t=o\left( n \ln\left(\frac{\phi_0}{ \sigma^2 n}\right) \right)$ we have  
$\E{\bar{\phi}(\Valvec{t})} = \omega( \sigma^2 n)$~. 	
   \end{enumerate}
  \end{theorem}
  
  While the above theorem shows a quick convergence to the running average, 
  this does not imply convergence to the initial average. In fact, as time progresses the distance to the initial average ($TSS(\Valvec{t})$) is likely to increase. 
  Nonetheless, in the case of the \Gaussian we can bound the drift of the running average from the initial average in a time window of $O(n^2)$ steps (cf. \autoref{lem:gaussian}). \autoref{thm:runningavg} roughly says that after at least $t=\Omega(n\log n)$  steps the distance to the running average is small  if we start with a potential that is polynomial in $n$. Thus, as long as $t=\Omega(n\log n)$ and $t=O(n^2)$ we  obtain $TSS(\Valvec{t})=  O\left( n  \right)$.    
  After the $O(n^2)$ step time window the potential starts to increase again, which, is unavoidable, due to the noise causing drift of the running average; in \Gaussian, the running average after $t$ steps diverges  with constant probability from the initial average by $\frac{ \sqrt{t } }{n}$ (\autoref{lem:gaussian}). This 
  in turn implies that $TSS(\Valvec{t}) \geq t/n$.

    \begin{corollary}[(Bounded) Divergence from Initial Avg.]\label{thm:initialavg}
      In the case of  \Gaussian, for any $\delta \in \mathbb{R}$ and large enough $n=n(\delta)$  and all  $t=\Omega\left( n \ln\left(\frac{\bar{\phi}(\Valvec{0})}{\delta \sigma^2  n}\right) \right)$ we have 
			\begin{enumerate}
			\item  `\emph{non-divergence} for $O(n^2)$ steps', i.e.,  $TSS(\Valvec{t})=  O\left( \left(\frac{t}{n} + n \right) \sigma^2  \ln(1/\delta) \right)$ with probability at least $1-\delta$  and 
			\item  `\emph{divergence} for $\omega(n^2)$ steps', i.e.,   $TSS(\Valvec{t})=  \Omega\left( \left(\frac{t}{n} + n \right) \sigma^2 \right)$ with constant probability.
			\end{enumerate}
  \end{corollary}
	If one can bound the divergence  between the running average and the initial average for a general noise-distribution $\aleph$ with (at least) exponential-decay\footnote{Again, we only require the function to be smooth, which we define in \autoref{sec:seq}.} the following remark is useful to obtain a similar bound for the $TSS(\Valvec{t})$ as in \autoref{thm:initialavg}.
 \fnote{give more details on how to prove the lower bounds}
  Recall that   $\avg{t}= \sum_i \Val{i}{t}/n$ and 
  in particular, $\avg{0}$ denotes the initial average. 
	\begin{remark}\label{remark:otherDistributions}
    Fix any $\delta \in \mathbb{R}$. Let $n=n(\delta)$ be large enough.
For any fixed $t=\Omega\left( n \ln\left(\frac{\phi_0}{\delta \sigma^2  n}\right) \right)$ with probability at least $1-\delta$ we have $TSS(\Valvec{t})= \Theta\left( n\left( \avg{t}-\avg{0} \right)^2+\sigma^2 n \ln(1/\delta) \right)$~. 
	\end{remark}	
\autoref{remark:otherDistributions} follows by rewriting  $TSS(t) = \bar{\phi}(\Valvec{t}) + n\cdot \left( \avg{0}-\avg{t} \right)^2 $ (cf. \autoref{lem:relTSS}) and plugging in the first part of  \autoref{thm:runningavg}. 
  \autoref{thm:initialavg} then follows by plugging in the bounded deviation of the running average from the initial average for the \Gaussian (cf. \autoref{lem:gaussian}).

\paragraph{The Influence of Rounding}
Agents with limited computational power might not be able to store real values. Motivated by this we also consider the setting where agents can only store integers. In particular, we consider the case that the averaging protocol is augmented with the following rounding procedure: Assume that the noise $N \sim \aleph$ takes only integer variables.
After a node $i$ receives the value from node $j$, the node averages it as before and then rounds up or down with equal probability. 
In \autoref{sec:rounding} we show how to relate the setting of rounding to the original setting allowing us to derive the following corollary. 
\begin{corollary}\label{thm:rounding}
The bounds of \autoref{thm:runningavg} and \autoref{thm:initialavg}  hold even if rounding is used.
\end{corollary}

\paragraph{The Synchronous Model}
In \autoref{sec:synch}, we show how our results extend to the synchronous setting. It turns out that the results are the same up to a rescaling of time.
\begin{corollary}[Synchronous Setting]\label{thm:synch}
The bounds of \autoref{thm:runningavg} and \autoref{thm:initialavg} hold even in the synchronous setting, where time is rescaled by a factor of $2/n$.
\end{corollary}

\paragraph{Experimental Results}
In \autoref{sec:simulations}, we simulate the averaging dynamic in various settings. In the first setting, we consider the distribution of the distances between agents' values and the running average. Our simulations show that these distances seem to follow an exponential law, i.e., the  concentration is even stronger than what \autoref{thm:runningavg} implies.

\smallskip

Due to the limited memory of agents it would be desirable to obtain similar results as in \autoref{thm:runningavg}  for the averaging dynamic in the setting where agents can only store values from a bounded range. However, our simulations in \autoref{sec:simulations} show that this setting  leads to a much faster (by order of magnitude) divergence between the running average and the initial average.

	
  \subsection{Technical Contributions}
%
%

  While it is not hard to show that in expectation the  potentials $TSS(t)$ and $\bar{\phi}(t)$ decrease in one step as long as their value is large, 
  it is surprisingly challenging to derive probabilistic bounds on either potential  at an arbitrary point in time, \ie
  bounds of the type $\Pr{\bar{\phi}(t) \geq b} \leq p(b)$.
  Two of the reasons are as follows. 
  (i) The potential decreases (expectedly) only conditioned on the fact that it is large enough.
  In fact, when the potential is small,  then due to the noise it will increase in expectation. 
  (ii) Since we study general distributions and in particular the normal distribution, the noise in a given round can be arbitrarily large leading to an arbitrarily large increase in $\bar{\phi}(t)$; if the protocol runs long enough (possibly exponentially long in $n$) we, indeed, will have encountered some time steps with a very large potential increase. 
  There are surprisingly few analytical tools for using potentials as $\bar{\phi}(t)$ with challenges (i) and (ii). 
 One notable exception is Hajek's theorem \cite{H82}, which can be used to bound the value of such a potential at a given time $t$. 
 However, in our setting---with our potential function---the results obtained are very weak.\footnote{Hajek's theorem  considers the moment generating function of the potential. In order to apply the theorem to our potential, it seems that one would need to consider a logarithmic version of the potential, which together with the moment generating function results in bound that is weaker than a simple union bound.}

Instead, we use a more sophisticated approach that at its core has a decomposition of the potential change in a single time step into three additive (but dependent) random variables. We iterate  this decomposition over time  throughout some interval $\mathcal{I}=(t_0,t_1]$  and sum the respective variables  which we will denote as $S^-(\mathcal{I})$, $S'(\mathcal{I})$, and $S^*(\mathcal{I})$. 
Then (cf. \autoref{lem:rewritingpotential}) we are able to bound the potential change at the end of the interval as
	\begin{align}\label{eq:needsomefoodhere}
\bar{\phi}(\Valvec{t_1}) &\leq   \left(1- \frac{ S^-(\mathcal{I}) }{t_1-t_0}\right)^{t_1-t_0} \cdot \bar{\phi}(\Valvec{t_0})  + S'(\mathcal{I}) + S^*(\mathcal{I}).
\end{align}  
Due to the dependencies between the three variables we use strong Martingale concentration bounds to separately upper bound $S'(\mathcal{I}) + S^*(\mathcal{I})$ and lower bound $S^-(\mathcal{I})$ (cf. \autoref{lem:bounds}). 
We then use union bound---to circumvent the dependencies---to bound each of these variables allowing us to get a bound on \autoref{eq:needsomefoodhere}.
It is critical that we define the random variable $S^-$ in such a way that it always has an expected decrease. This is in stark contrast to the entire potential, which, as we mentioned before in (i), only decreases in expectation when it is large. Having an unconditional decrease of $S^-$ allows us to consider arbitrarily large intervals.  
With these bounds at hand one can use \autoref{eq:needsomefoodhere} to obtain probabilistic bounds on the potential at any given point time $t_1$.  However, due to the  bound on  $S'(\mathcal{I}) + S^*(\mathcal{I})$ the total bound becomes very weak for large intervals.  
As a remedy, we carefully trace the change in the potential in different regimes (with several phases in each regime) and we separately apply the aforementioned analysis with a fresh (small) interval in each phase. The intervals (and thus also the phases) have variable length---decreasing geometrically or even exponentially, depending on the regime.  


\vspace{-0.3cm}
\section{Model}\label{sec:model}
\vspace{-0.1cm}

  In this section we present the model including all assumptions.  
  We have a collection of $n$ agents that have initial values    $\Val{1}{0}, \Val{2}{0},\dots, \Val{n}{0}$. 
  Time is discrete and $\Val{i}{t}$ denotes  the value of agent $i\in [n]$ at time $t$. Recall that $\avg{t}= \sum_i \Val{i}{t}/n$ denotes the average value at time $t$; in particular, $\avg{0}$ denotes the initial average.  For two random variables $X$ and $Y$ we write $X\stackrel{d}{=} Y$ if they have the same (probability) distribution.  Next, we define the communication models.
  \begin{definition}[Communication Models]\label{def:comm}
  We consider two communication models. 
  \begin{enumerate}
\fnote{pic one}
  \item \emph{Sequential model}: At every discrete time step two of the agents $i,j$ are chosen uniformly at random (with replacement\footnote{This is not crucial to our results, but simplifies the calculations slightly.}) and exchange their current values $x_i$ and $x_j$, where the values received are  $x_i + N_i$ and $x_j + N_j$, where $N_i,N_j \stackrel{d}{=} N$.
  \item  \emph{Synchronous model}: At every discrete time step a perfect matching is chosen u.a.r. among all perfect matchings on the $n$ agents.\footnote{Again, we allow matchings of the kind $(i,i)$ for simplicity. It is easy but slightly less aesthetic to modify our results to exclude matchings $(i,i)$. } All matched agents interchange their values as in the sequential model.

  \item \emph{Sequential model}: At every discrete time step two of the agents $i,j$ are chosen uniformly at random (with replacement\footnote{This is not crucial to our results, but simplifies the calculations slightly.}) and send their current values $x_i$ and $x_j$ to each other, where the values received are  $x_i + N_i$ and $x_j + N_j$, where $N_i,N_j \stackrel{d}{=} N$.
  \item  \emph{Synchronous model}: At every discrete time step a perfect matching is chosen u.a.r. among all perfect matchings on the $n$ agents\footnote{Again, we allow matchings of the kind $(i,i)$ for simplicity. It is easy but slightly less aesthetic to modify our results to exclude matchings $(i,i)$. }. All matched agents interchange their values as in the sequential model.

  \end{enumerate}
We use the \emph{parallel time}, which was first defined in \cite{Alistarh}, to denote the time step $t/n$ in the sequential model. This notion eases the comparison of results in both models, as the total number of interactions is up to a factor of $2$ equal.
  \end{definition}

  \begin{definition}[Noise Models] \label{def:noiseModels}
Let $v$ be the value sent by an agent. The value received is  $v+N$, where
  $N$ is distributed according to some zero-mean noise distribution $\aleph$ and let $\sigma^2 =\Var{N}$. 
  \end{definition} 
	  We consider general noise distributions and our results depend on the moments of $N$. The following two models are of special interest in this paper.
  \begin{enumerate}
  \item \emph{Gaussian white noise model} where $\aleph= \mathcal{N}(0,\sigma^2)$ for an arbitrary $\sigma$.    
  \item \emph{Discrete white noise model} where $\aleph = \mathcal{D}(p)$, with $\Pr{ N = i } = \frac{1}{2} p(1-p)^{|i|} $, for $i\in \mathbb{Z}\setminus \{ 0\}$ 
   and $\Pr{ N = 0} = p$, where  $p\in (0,1]$. Note that $\Var{N}=\frac{1-p}{p^2}$.
  \end{enumerate}
From now on we assume that the noise $N$ is distributed according to a fixed noise distribution $\aleph$ that is independent of $n$.

  \begin{definition}[Averaging Dynamic]\label{def:alg}
  We consider the real valued and the discrete valued algorithm.
A node with value $v$ at time receiving the input $w$ sets its new value to
\begin{enumerate}
\item $v'=\sfrac{(v+w)}{2}$ in the \emph{real valued model}.
\item  $v'=\begin{cases}
 \ceil{\sfrac{(v+w)}{2} } & \text{ w.p. $\frac{1}{2}$}\\
 \floor{\sfrac{(v+w)}{2} } & \text{ otherwise}
 \end{cases}$ in the \emph{discrete valued model}. 
\end{enumerate}

  \end{definition} 


  	A probability distribution $\mathcal{D}$ is called \emph{sub-Gaussian} if for $X\sim \mathcal{D}$ we have that there exists positive constants $c_1,c_2$ such that for every $x$ we have $\Pr{|X| \geq x} \leq c_1 \exp(-c_2 x^2 ).$

Whenever we calculate the new values  $\Valvec{t+1}$ by
conditioning on the current state, $\Valvec{t}=\valvec{t}$ we use small letters $\val{i}{t}$ to denote fixed values and capitalized letters $\Val{i}{t+1}$ to denote random variables. Furthermore, we use bold-face to denote vectors.
Throughout the paper we will assume that the number of agents $n$ is large enough and in particular $n\E{N^2}\geq 1.$

We define the following potentials which are essential in all our proofs and formal results.
\begin{definition}[Potentials]
\[ TSS(\valvec{t}) = \sum_{i}\left( \val{i}{t} - \avg{0} \right)^2, \hspace{0.2cm}
 \bar{\phi}(\valvec{t}) =  \sum_{i}\left( \val{i}{t} - \avg{t} \right)^2, \hspace{0.2cm} 
   \phi(\valvec{t}) = \sum_{i,j}\left( \val{i}{t} - \val{j}{t} \right)^2   .\]
\end{definition}
When clear from the context we drop the time index $t$ and we write $\mathbf{x}$ instead of $\valvec{t}$, $x_i$ instead of $\val{i}{t}$, etc. Similarly we will use the following short forms $TSS(t)=TSS(\valvec{t})$ and $\bar{\phi}(t)=\bar{\phi}(\valvec{t})$.
We emphasize that the difference between $\bar{\phi}(\mathbf{x})$ and $TSS(t)$ is that the former measures the squared distance w.r.t.  the \emph{running average} and the latter w.r.t. \emph{initial average}.
Initially, we have $\bar{\phi}(\valvec{0})=TSS(0)$.
The following fact  shows how $\bar{\phi}(\Valvec{t})$ relates to $TSS(t)$ and how $\bar{\phi}$ relates to $\phi$.

\begin{fact}\label{lem:relTSS}
We have that
\begin{enumerate}
\item $TSS(t) = \bar{\phi}(\Valvec{t}) + n\cdot \left( \avg{0}-\avg{t} \right)^2 $ and
\item $\phi(\mathbf{x})     = 2n \cdot  \bar{\phi}(\mathbf{x})  . $
\end{enumerate}
\end{fact}
\begin{proof}
Consider part $(i)$.
\begin{align*}
	TSS(t)&= \sum_i \left( \val{t}{i}- \avg{0} \right)^2 =
	\sum_i \left( x_i-\avg{t} + \avg{t}-\avg{0} \right)^2\\
	&= \sum_i \left(\left( \val{t}{i}-\avg{t} \right)^2+ 2(\val{t}{i}-\avg{t}) (\avg{0}-\avg{t} )+\left( \avg{0}-\avg{t} \right)^2 \right)\\
	&= \bar{\phi}(\Valvec{t}) +  2\left(\sum_i \val{t}{i} - n\avg{t}\right)(\avg{0}-\avg{t})  + n\left( \avg{0}-\avg{t} \right)^2\\
	&= \bar{\phi}(\Valvec{t}) + n\left( \avg{0}-\avg{t} \right)^2. 
\end{align*}
\end{proof}
\medskip
Consider part $(ii)$.
\begin{align*}
\phi(\mathbf{x}) &= \sum_{i,j}\left( x_i - x_j \right)^2 = 2 n \sum_{i}x_i^2 -2 \sum_{i,j}x_i x_j = 2 n \sum_{i}x_i^2 - 2 n \varnothing \sum_{i}x_i
\notag\\
&= 2 n \left( \sum_{i}x_i^2  -  \sum_i x_i \varnothing \right)      =  2 n \left( \sum_{i}x_i^2  - 2 \sum_i x_i \varnothing +  n\varnothing^2\right)     =2 n \sum_{i}\left( x_i - \varnothing \right)^2 =  2n \cdot \bar{\phi}(\mathbf{x})  . 
\end{align*}


Note that many alternative ways to define the potential at a time $t$  such as the max distance and $\ell1$ norm give only a very partial picture:
The max distance to the mean for example does not distinguish between just one node being far and all nodes being far. On the other hand, the $\ell1$ norm does not 
does not `punish' outliers enough: there is no difference between $n$ nodes being off by $1$ from the average and one node being off by $n$. 
\paragraph{Notation}\label{sec:expavg}
\label{sec:additional}
We use $X\sim \mathcal{D}$ to denote that $X$ is distributed according to probability distribution $\mathcal{D}$.
For two random variables $X$ and $Y$ we write  $X \leq^{\text{st}} Y$ if $X$ is  \emph{stochastically dominated} by $Y$, i.e., $\Pr{X \geq x} \leq \Pr{Y \geq x}$ for all $x\in \mathbb{R}$. 
 We use
$\norm{\mathbf{x}}_2$ to denote the $L2$-norm.   
  In the sequential model we have two random variables $N_1^{(t)}$ and $N_2^{(t)}$ for the noise of the channel at time step $t$ (recall that $N_1^{(t)}$ and $N_2^{(t)}$ are distributed according to $\aleph$). We define the following two random variables $N'^{(t)}$ and $N^{*(t)}$ that  will play a key-role in our analysis:
	\[N'^{(t)}=\left(N_1^{(t)}\right)^2 + \left(N_2^{(t)}\right)^2, \hspace{1.7cm} N^{*(t)}=N_1^{(t)} + N_2^{(t)}~.\]
	\begin{fact}
	In the Gaussian noise model, we have
$N^{*(t)}\sim \mathcal{N}(0,2\sigma^2)$ and $N'^{(t)}\sim \Gamma(1,2\sigma^2)$, where $\Gamma(\cdot,\cdot)$ denotes the gamma distribution.
	\end{fact}
	When clear from the context we simply write $N'$ and $N^*$  instead of $N'^{(t)}$ and $N^{*(t)}$, respectively.
We use $\mathcal{F}_t$ to denote the filtration at time $t$, which encapsulates all randomness up to time $t$ as well as the initial values of the nodes; hence it defines the state at time $t$ completely.

\section{The Sequential Setting: Convergence towards the Running Average}\label{sec:seq}
Conditioning on all the randomness until time $t$, \ie conditioning on $\mathcal{F}_t$, we define 

$\Delta^{(t+1)}  = \begin{cases}
\frac{\left(\val{i}{t}-\val{j}{t}\right)^2}{2\bar{\phi}(\valvec{t})} & \text{ for $\bar{\phi}(\valvec{t}) > 0$}\\
1/n & \text{ otherwise }
\end{cases} $, where $i$ and $j$ are the chosen in round $t$.
\begin{lemma}[One Step Bound]\label{lem:exactdistribution}
Fix an arbitrary potential at time $t$. Suppose the pair $i,j$ was chosen to communicate and
condition on the filtration $\mathcal{F}_t$ (all events that happened up to round $t$). Then, the following holds
\[ \bar{\phi}(\Valvec{t+1})-\bar{\phi}(\valvec{t})  \fnote{\leq^{\text{st}}} \leq   - \Delta^{(t+1)}\bar{\phi}(\valvec{t})+  \frac{{N'^{(t+1)}}}{4}+ N^{*(t+1)} \left(\frac{\val{i}{t}+\val{j}{t}}{2} - \avg{t} \right)~. \]
Further we have $\E{\Delta^{(t+1)}~|~\mathcal{F}_t} = \frac{1}{n}$ . 
\end{lemma}
In order to prove the statement, we first calculate the exact expected change in one step (\autoref{lem:onestep}). We then majorize (stochastic dominance) with the slightly more convenient statement above.

For an arbitrary time interval $\mathcal{I}$ define 
\[S'(\mathcal{I})= \sum_{\tau \in \mathcal{I}}  N'^{(\tau)}/4, \hspace{0.4cm} S^*(\mathcal{I})= \sum_{ \tau \in \mathcal{I}} N^{*{(\tau)}}  \left(\frac{\val{i}{\tau-1}+ \val{j}{\tau-1} }{2} - \avg{\tau} \right), \hspace{0.4cm}S^-(\mathcal{I})= \sum_{ \tau \in \mathcal{I}} \Delta^{(\tau)}~.\]
Note that,  in the definition of $S^*$, we sum up over all time steps $\tau$ in the interval $\mathcal{I}$ and we consider the pair
 $i$ and $j$ that  is chosen in round $\tau$ (in each round a different pair $i$ and $j$ can be chosen). With \autoref{lem:exactdistribution} and the definitions of $S', S^*$ and $S^-$ we can deduce the following decomposed bound on the potential for an arbitrary interval.
\begin{proposition}[Decomposition of Potential]\label{lem:rewritingpotential} Fix  arbitrary $t_0, t_1$ and
consider the interval $\mathcal{I} = (t_0,t_1]$. For $t=t_1-t_0$ we have that
\begin{align}
\bar{\phi}(\Valvec{t_1}) &\leq   \left(1- \frac{ S^-(\mathcal{I}) }{t}\right)^t \bar{\phi}(\Valvec{t_0})  + S'(\mathcal{I}) + S^*(\mathcal{I}). 
\tag{\ref{eq:needsomefoodhere}}
\end{align} 

\end{proposition}

In the following we define smooth noise distributions. Define
\[ m_{t,\delta} =\arg\max_\ell \left\{    \Pr{\max\left(\left\{N^{'(t_0)}, \dots , N^{'(t_0+t)} \right\}  \cup \left\{N^{*(t_0)}, \dots , N^{*(t_0+t)}  \right\} \right) \leq \ell } \geq 1-\delta   \right\} .\]\dnote{something wrong with the notation here: is $\{N^{'(t_0)}, \dots , N^{'(t_0+t)} \leq \ell \} $ a set?}
\fnote{fixed?}
 Using strong martingale concentration bounds (\autoref{pro:fancyazuma} and \autoref{pro:fancyazumalower}) and bounding the variance, we deduce the following upper bound on $S^*+S'$ and lower bound on $S^-$. 


\begin{lemma}
Let $t_0,t_1$ be such that $t_1 > t_0$ and consider the interval $\mathcal{I}=(t_0,t_1]$.
\begin{enumerate}\label{lem:bounds}
\item  With probability $1-\delta$ we have
\begin{align*}
S^*(\mathcal{I}) + S'(\mathcal{I}) \leq \hspace{11cm}\\
\hspace{1cm}  \frac{t}{4}\E{N'} +\ 5\sqrt{\frac{t}{n}} \left(\ln(4t/\delta) m^*_{t,\delta/4}\right)^2 \left(2+ \E{N'} \right)   \sqrt{  \bar{\phi}(\valvec{t_0}) + 9 {t} \E{N'}    +2 }~.
\end{align*} \label{lem:combinedugliness}
\item For any $\gamma < 1$, w.p. at least $1- \exp\left(-\frac{3\gamma ^2t}{8n}\right)$ we have
$ S^-(\mathcal{I}) \geq (1-\gamma)\frac{t}{n}.$\label{lem:sminus}
\end{enumerate}
\end{lemma}

Our main results only hold for smooth noise distributions, which we define in the following.

\begin{definition}\label{def:smooth}
A noise distribution $\aleph$ is \emph{smooth} if for all  $\delta >0$ and all $t > 0$  we have
 $m_{t,\delta} \leq \left(\frac{t}{\delta}\right)^{1/20}$. 
 \end{definition}

 However, note that any (sub-)linear probability distribution and even some inverse polynomial distributions are smooth. Thus many practically relevant distributions such as Gaussian, binomial and Poisson distributions are smooth. For example, for the standard normal distribution ($N\sim \mathcal{N}(0,1)$) we have $m_{t,\delta}=\log(t/\delta)$, since in each time step the probability that the $N^2$ exceeds $\log(t/\delta)$ is equal to the probability that $N$ exceeds $\sqrt{\log(t/\delta)}$ which happens w.p. at most $\delta/t$. Taking union bound over all $t$ steps shows that it is smooth.

For smooth noise distributions we can upper bound the additive increase due to the noise. 

The following proposition almost directly implies \autoref{thm:runningavg}.
\begin{proposition}\label{pro:main}
Fix any $\delta \in (0,1]$ and assume that the noise distribution is smooth.
There exists a constant $c$ such that for a time step $t_0$ with potential  $\bar{\phi}(\valvec{t_0} ) $ we have
\[ \Pr{  \bar{\phi}(\Valvec{t^*} ) \geq   \ln(1/\delta) n \E{N'} + b  ~|~\mathcal{F}_{t_0} } \leq  \delta,   \]
where $t^* =t_0 + c n \ln\left( \frac{ \bar{\phi}(\valvec{t_0} ) }{ \E{N'} n \delta}\right) $ and $b=2   \left(1+ \E{N'} \right)\left(\ln(1/\delta) \right)^{9} n^{9/10} .$
\end{proposition}
\begin{proof}[Proof Sketch]
 We only sketch the proof idea for a simplified setting; during the sketch we assume that  $N\sim \mathcal{N}(0,1)$ (with $\E{N'}=O(1)$) and also that $\delta$ is at least $1/n^{3}$. 
The main ingredients for the proof are \autoref{lem:rewritingpotential} and  \autoref{lem:bounds}. For an interval $\mathcal{I}=(t_0,t_1]$ \autoref{lem:rewritingpotential} upper bounds the potential at time $t_1$ by
\begin{align}
\label{eqn:decomposition}
\bar{\phi}(\Valvec{t_1}) &\leq   \left(1- \frac{ S^-(\mathcal{I}) }{t}\right)^t \bar{\phi}(\Valvec{t_0})  + S'(\mathcal{I}) + S^*(\mathcal{I}),
\end{align} 
where $t$ is the length of the interval.  \autoref{lem:bounds} lower bounds $S^-(\mathcal{I})$ and  upper bounds the sum $S'(\mathcal{I}) + S^*(\mathcal{I})$. To prove \autoref{pro:main} we have to show that the initial potential $\bar{\phi}(\valvec{t_0} ) $ decreases to $O(n)$ after $O(n\cdot \log \bar{\phi}(\valvec{t_0}))$ time steps with probability $1-\delta$. Optimally, we would use a single application of \autoref{lem:rewritingpotential} to upper bound the potential as in \autoref{eqn:decomposition} and then bound the terms $S^-(\mathcal{I})$ and $S'(\mathcal{I}) + S^*(\mathcal{I})$ via  \autoref{lem:bounds}. However, the bounds on $S^-$ and $S'+S^*$ given by  \autoref{lem:bounds} are too loose to yield the desired result via a single application of \autoref{lem:rewritingpotential} and  \autoref{lem:bounds}  with the whole time interval $\mathcal{I}=[t_0,t_0+O(n\log \bar{\phi}(\valvec{t_0}))]$. 
For example, the bound on $S'+S^*$ inherently has a term of order $\sqrt{\bar{\phi}}$, where $\bar{\phi}$ is the potential at the start of the interval for which  \autoref{lem:bounds}, \ref{lem:combinedugliness} is applied. Thus a one shot proof as described above can never reach a potential below $\sqrt{\bar{\phi}}$. This is not sufficient if the initial potential is large, e.g., say for $\bar{\phi}\gg n^{8/3}$. 

To circumvent this problem we apply \autoref{lem:rewritingpotential} and \autoref{lem:bounds} several times for smaller time intervals: More detailed, we split the proof of \autoref{pro:main} into two regimes. In regime $2$ we use several phases to decrease the potential to $\Theta(n^{4/3})$. If the potential is $\bar{\phi}$ at the beginning of a phase a single application of \autoref{lem:rewritingpotential} and  \autoref{lem:bounds}  reduces the potential to $\bar{\phi}^{3/4}$. The length of each such phase is geometrically decreasing by a factor $3/4$ where the first phase is of length $O\left(n \ln\left( \frac{ \bar{\phi}(\valvec{t_0} ) }{ n \delta}\right)\right) $. After the last phase of regime $2$ the potential is of order $n^{4/3}$.  

Then, in regime $1$ the potential reduces from $\Theta(n^{4/3})$ to $O(n)$, again through several phases. 
If the first phase of regime 1 starts with a potential of size  $B$, the phase has length $t=O(n\ln(B))$. If there was no additive increase due to the noise, then this would reduce the potential to $0$. However, there is an additive increase of $\Theta(t)=\Theta(n\ln(B))$ which leaves us with a potential of size $O(n\ln(B))$. The next phase will therefore be of length $n\ln\ln(B)$ etc.
This is repeated for $\ln^*(B)$ phases until the potential reduces to $O(n)$, which, as we explained in \autoref{sec:formalResults}, is the furthest the potential can be decreased .

Putting everything together, we get that after
$O\left(n \ln\left( \frac{ \bar{\phi}(\valvec{t_0} ) }{ n \delta}\right)\right) $ rounds the potential reduces to $O(n)$. 
\end{proof}
The full proof of \autoref{pro:main} handles general $\E{N'}$ and general $\delta$ and thus it is significantly more technical. It can be found in \autoref{sec:mainPropositionProof}. From \autoref{pro:main} we are able to derive \autoref{thm:runningavg}, whose proof can be found in the appendix.

%
\section{Deviation from the Initial Average}\label{sec:devinitial}
An informal argument for the statements in this section in the special case of $\sigma =1$ can be found in \cite{xiao2004fast}.
Before we state our results we need the following result on the standard normal distribution.
\begin{theorem}[\cite{johncook}]\label{thm:cook}
Let $\Phi(x)$ denote the cumulative distribution function of the standard normal distribution.
We have for $x\geq 0$: 
\[
\frac{1}{\sqrt{2\pi}} \frac{x}{x^2+1}  \exp\left( -  x^2/2  \right)  \leq \Phi(x) \leq  \frac{1}{\sqrt{2\pi}}\frac{1}{x} \exp\left( -  x^2/2  \right). 
\]
\end{theorem}
We can now state and prove the main results of this section.
\begin{lemma}\label{lem:gaussian}
For any $t$ and any  $\delta < 1$ , we have
$  \avg{t}-\avg{0} \sim \frac{\sum_{\tau = 1}^{2t} N^{(\tau)}}{2n} $
with probability at least $1-\delta$, where $N^{(\tau)}$ is the  noise of the channel. 
In particular, for the \Gaussian setting where $N\sim \mathcal{N}(0,\sigma^2)$ we have $\sum_{\tau=1}^{2t} N^{(\tau)} \sim \mathcal{N}(0,2t \sigma^2)$.
Thus 
\begin{enumerate}
\item $ |\avg{t}-\avg{0}| \leq  \frac{ \sigma\sqrt{t \ln(1/\delta)} }{n}\text{  w.p. at least $1-\delta$} $
\item  $ |\avg{t}-\avg{0}| \geq \frac{\sigma \sqrt{t \ln(1/\delta) }}{n} \text{  w.p. at least $\frac{\delta}{2\sqrt{2\ln(1/\delta)}}.$} $ \end{enumerate}
\arxiv{Moreover, for the \Disc with $n =\Omega(1/p^2)$. \fnote{finish this} with $\sigma = \sqrt{1-p}/p$}
\end{lemma}
\begin{proof}
Note that $\mathbf{1}^T\Valvec{t+1}=\frac{N_i+N_j}{2} + \mathbf{1}^T\Valvec{t}$, where $i$ and $j$ are the nodes scheduled in the current round and $\mathbf{1}^T\Valvec{t} = \sum_i \val{i}{t}$.
Applying this recursively and using that all $N_i$ follow the same distribution we have
$\mathbf{1}^T\Valvec{t}= \frac{\sum_{\tau=1}^{2t} N^{(\tau)}}{2}+ \mathbf{1}^T\Valvec{0}$. 
Using that $\avg{t}=\mathbf{1}^T\Valvec{t}/n$ completes the proof of the first part.

Consider $(a)$. 
For a general normal distribution with mean $\mu_x$ and variance $\sigma^2_x$ we have that
\[\Pr{ 2n(\avg{t}-\avg{0})  \geq  x  } = \Phi\left( \frac{x-\mu_x}{\sigma_x} \right)~,\] where $\Phi(x)$ denotes the cumulative density function of the standard normal distribution.
Applying the upper bound of \autoref{thm:cook} and using symmetry of the normal distribution,  it holds for  $x=2\sigma  \sqrt{t \ln(1/\delta)}$, $\sigma_x=\sqrt{2t\sigma^2}$ and $\mu_x = 0$ that
 \begin{align*}
\Pr{ 2n(\avg{t}-\avg{0})  \leq  -x  } &=\Pr{ 2n(\avg{t}-\avg{0})  \geq  x  } \leq \frac{\sigma_x}{(x-\mu_x)\sqrt{2\pi}} \exp\left( -  \frac{(x-\mu_x)^2}{2\sigma_x^2}  \right)\\
&\leq \frac{\sqrt{2t\sigma^2}}{2\sigma \sqrt{t \ln(1/\delta)} \sqrt{2\pi}} \exp\left( -  \frac{4\sigma^2 t\ln(1/\delta)}{4t\sigma^2} \right) \leq \frac{\delta}{2}.
\end{align*}  
Taking Union bound yields the claim.

Consider $(b)$.
By applying the lower bound of \autoref{thm:cook} and using similar arguments as before, we have for $x=2\sigma  \sqrt{t \ln(1/\delta) }$, $\sigma_x=\sqrt{2t\sigma^2}$, $\mu_x = 0$
and  $y=\frac{x-\mu_x}{\sigma_x}= \sqrt{2\ln(1/\delta)}\geq \sqrt{2}$
 \begin{align*}
\Pr{ 2n(\avg{t}-\avg{0})  \leq  -x  } &=\Pr{ 2n(\avg{t}-\avg{0})  \geq  x  } \geq \frac{1}{\sqrt{2\pi}} \frac{y}{y^2+1} \exp\left( -  \frac{(x-\mu_x)^2}{2\sigma_x^2}  \right)\\
&\geq \frac{1}{2\sqrt{\pi}} \frac{y}{y^2}   \exp\left( -  \frac{4\sigma^2 t\ln(1/\delta)}{4t\sigma^2} \right) \geq \frac{\delta}{2\sqrt{2\ln(1/\delta)}}. & \qedhere
\end{align*}  
\end{proof}
%
%
%
%
Using the Berry-Esseen theorem, one can easily prove similar bounds for any distribution with bounded third moment including \emph{discrete white noise}. Similarly, rounding can easily be taken care of by applying the ideas from \autoref{sec:rounding}.

In the following we consider the potential $(\varnothing_t)_{t\geq 0}$  as a Martingale allowing us to use \autoref{pro:fancyazuma} to derive the desired concentration bounds. The following bound is weaker than the aforementioned bounds, however, it is useful 
whenever the noise is such that $m_{t,\delta/(2t) }   $ is small.

\begin{proposition}\label{lem:noname} 
For any $t\geq 2$ and any  $\delta < 1,$ we have 
$- m_{t,\delta/(2t) }   \sigma \sqrt{2t } \leq  \avg{t}-\avg{0} \leq   m_{t,\delta/(2t) }   \sigma \sqrt{2t } $
with probability at least $1-\delta$.

\end{proposition}
\begin{proof}
We start by showing that the sum of entries $(\mathbf{1}^T\Valvec{t})_{t\geq 0}$  is a Martingale
	\begin{align*}
		\E{\mathbf{1}^T\Valvec{t+1}- \mathbf{1}^T\mathbf{x} ~|~ \Valvec{t}=\mathbf{x}, \zeta_t = (i,j)  } &= 
		 \E{\frac{x_i+x_j+N_i}{2}}+\E{\frac{x_i+x_j+N_j}{2}}-x_i-x_j\\
		 &= \frac{\E{N_i}+\E{N_j}}{2} = 0.\\
\end{align*}
By law of total expectation, summing over all choices of $i$ and $j$, we get that $(\mathbf{1}^T\Valvec{t})_{t\geq 0}$  is a Martingale.
%
%

Since $(\mathbf{1}^T\Valvec{t})_{t\geq 0} $ is a Martingale, so is $(\varnothing^{(t)})_{t\geq 0}$, where we used that $\varnothing^{(t)}=\mathbf{1}^T\Valvec{t}/n$.
Note that \[ -m_{t,\delta/(2t)} \leq \mathbf{1}^T{\Valvec{t+1}}- \mathbf{1}^T{\Valvec{t}} \leq m_{t,\delta/(2t)}\] w.p. at least $1-\delta/(2t)$ per time step and hence, by Union bond, w.p. at least $1-\delta/2$ throughout the interval. 
By \autoref{pro:fancyazuma}, with $M=m_{t,\delta/(2t)}$, $\sigma_i^2 \leq \sigma^2$, and  $b= m_{t,\delta/(2t) }   \sigma \sqrt{2t }$ we get 
\[ \Pr{|\avg{t}-\avg{0}| \geq b}
 \leq \exp\left(-\frac{b^2}{2\left(\sum_{i=1}^t
		\sigma^2  +  Mb/3\right)}\right)
 \leq \exp\left(- \ln(2t/\delta) \right)  . \]
 Taking Union bound yields the r.h.s. inequality of the claim
 The l.h.s. follows by using \autoref{pro:fancyazumalower} instead of \autoref{pro:fancyazuma}. 
\end{proof}



\section{Experimental Results}\label{sec:simulations}
The goal of this section is twofold. First, we seek to better understand the distribution $\mathcal{D}$ of the distances $\val{i}{t}-\avg{t}$. Second, we simulate a setting in which the range of values is bounded motivated by computational and storage limited agents. 
All results in this section are based on an implementation of the simple averaging dynamic.
The code (python3) for the experiments can be found here \cite{ourcode}.

\begin{figure}[ht!]
\centering
\begin{subfigure}[t]{.47\textwidth}
		\includegraphics[width=0.99\textwidth]{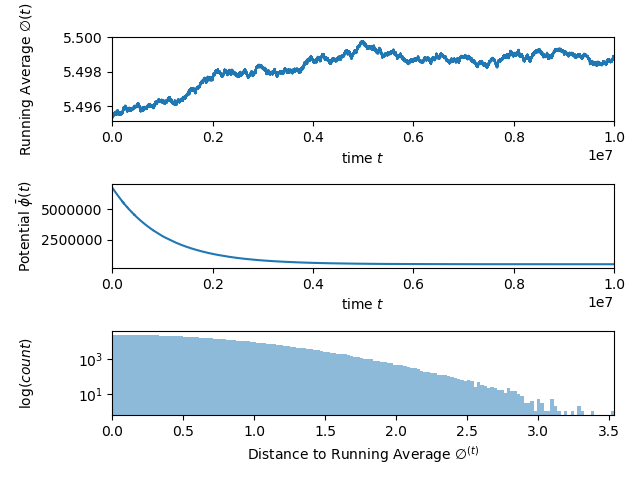}
	\caption{The setting of this example is: $n=10^6$, initial distribution of values is uniformly at random in the range $[1,n^2]$, $10n$ iterations, Gaussian white noise with variance $1$, unbounded range.}
	\label{fig:histoGramm}
\end{subfigure}%
\hspace{0.3cm}
\begin{subfigure}[t]{.47\textwidth}
		\includegraphics[width=0.99\textwidth]{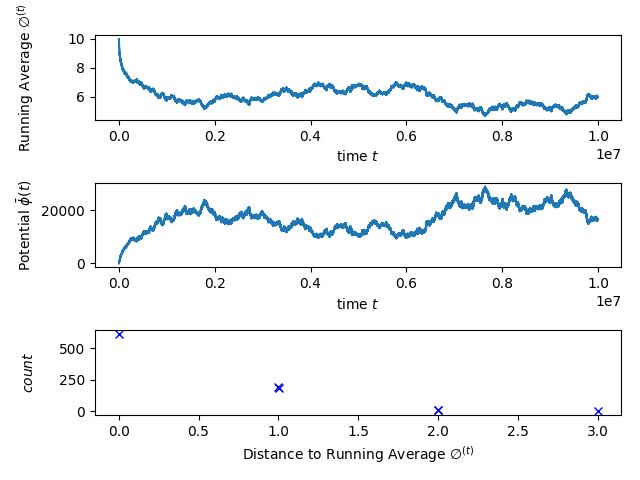}
	\caption{The setting of this example is: $n=1000$, all values equal to $10$, using discrete white noise model  $\mathcal{D}(0.8)$ (see \autoref{def:noiseModels}), bounded range in the interval $[1,10]$, $10^4n$ iterations. The avg. of the values drifts from $10$ to  $~6$.}
	\label{fig:ex2}
\end{subfigure}
\caption{The figure depicts the distribution of distances as well as the bounded value setting.}
\label{fig:test}
\end{figure}

\subsection{The Distribution of the Distances} 
The experiments suggest that the distance decays at least exponentially. Note that the experiments only show a single iteration, however, this phenomena was observable in every single run.
The bound on $\E{\bar{\phi}\left(\Valvec{t}\right)}$ we obtained in \autoref{thm:runningavg} only implies that $\mathcal{D}$ is at most $O(1/d^3)$. However, we conjecture, for sub-Gaussian noise that $\Pr{ |\Val{i}{t}- \avg{t}| \geq x} =O(\exp^{-x})$ (cf. \autoref{fig:histoGramm}). 
Showing this rigorously is challenging due to the dependencies among the values.
Nonetheless, such bounds are very important since they immediately  bound the maximum difference and we consider this the most important open question.
%

%

\subsection{The Bounded Values Setting}
One of the motivations for the very simple averaging dynamic arises in the setting of limited computational power of the interacting agents. So far we assumed that  agents can store and transmit (intermediate) values from an unbounded range. 
For many applications and in particular motivated by agents with bounded memory one would hope for similar results if there is a maximum and a minimum value that can be stored or transmitted.  
The formal definition is as follows: values can only be from the  range $[v_{min},v_{max}]$ ($=[1,10]$ in our experiments).
We assume noise of the channel cannot produce values larger than $v_{max}$ or smaller than $v_{min}$, which can be motivated as follows in the setting where the values correspond to amplitudes: 
here $v_{max}$ and $v_{min}$ are simply the amplitudes (high amplitude and no amplitude)  where the signal-to-noise ratio is very large, and noise becomes negligible. 
An equivalent model is that the agents  know the range of possible communication values, and hence,  they can simply correct every value larger than $v_{max}$ to $v_{max}$. In particular when agents only have limited storage, the communication range will often be bounded, and even rounding might become necessary (see \autoref{sec:rounding}). 

We refer to these equivalent models as the model with \emph{cutoffs}. 
While the experiments indicate that values still converge towards the running average, there is a clear drift of the running average from the initial average if the input values are chosen unsuitably. 
%
In our experiments, we set the range of values to $[1,10]$, use the noise described in the discrete noise model together with rounding. Initially, all agents have value $10$.
We see a drastic drift of the running average (see \autoref{fig:ex2}). Even though the initial average  is $10$,  the running average appears to approach the midpoint of the range, i.e., 5. The histogram of distances to the initial average shows even more clearly that the values are not concentrated around the initial average.
Although the experiments only show a single iteration, this phenomena was observable in every single run.
We believe that the reason for this is simply that the noise is no longer symmetric and no longer  zero-mean due to the cutoffs  $[1,10]$. Proving convergence to the running-average in this model seems challenging and interesting.


 
 We believe that the insights in bounding this potential might be useful in similar problems.

\section{Conclusion and Open Problems}

In this paper we showed bounds on the convergence time for the unbounded setting.
Our simulations in  \autoref{sec:simulations} yield two interesting open problems:
(i) study the setting where the values are restricted to some interval (in this case the noise is no longer symmetrical) and (ii) prove tail bounds on the distance distribution w.r.t. to the running or initial average. 
Another interesting research direction is to move away from zero-mean noise and consider biased noise models: how quickly can the bias(es) be estimated and is convergence still feasible by compensating for the (learned) bias?


\bibliography{biblio}
\newpage
\appendix

\section{Auxiliary Claims}
\begin{theorem}[Weierstrass Product Inequality]\label{thm:weierstrass}
We have
\begin{enumerate}
\item \[ \prod_i (1-x_i)^{w_i} \geq 1-\sum_i w_i x_i  , \] if $x_i \leq 1$ and either $w_i\geq 1$ for all $i$ or $w_i\leq 0$ for all $i$.
\item  \[ \prod_i (1-x_i)^{w_i} \leq 1-\sum_i w_i x_i ,  \] if  $\sum_i w_i \leq 1$,  $w_i \in [0,1]$ and $x_i \leq 1$ for all $i$. 
\end{enumerate}
\end{theorem}
\begin{proof}
Consider $(i)$ which trivially holds for $w_i \geq 1$. \arxiv{consider smaller zero case}
Now consider $(ii)$.
Taking the logarithm on both sides and treating the $w_i$ as probabilities, where we introduce a dummy element with $x_0=0$, $w_0 = 1-\sum_{i\geq 1} w_i$ and derive
  \[ \sum_i w_i \ln(1-x_i) \leq  \ln\left(  \sum_i  w_i \cdot (1-x_i) \right) \leq   \ln\left( 1-  \sum_i  w_ix_i \right),  \] 
Where we used Jensen's inequality.
This concludes the proof.

\end{proof}

\begin{proposition}[Distribution Facts]\label{lem:randomfacts}
Let $X^2 \sim \mathcal{N}(0,\sigma^2)$. We have
\begin{enumerate}
\item $\E{X^2}= \sigma^2$
\item $\Var{X^2}=2\sigma^4$.
\end{enumerate}
\end{proposition}
\begin{proof}
First observe that $X^2 \sim \sigma^2 \chi_1^2$, where $ \chi_1^2$ is the chi-squared distribution with $1$ degree of freedom. Hence, 
$\E{\chi_1^2} = 1$ and $\Var{\chi_1^2} = 2$ implying $(i)$ and $(ii)$.
\end{proof}

We will make use of a
slightly generalized version of the Hoeffding bound (see~\cite{H63}).

\begin{theorem}[\cite{H63}] \label{pro:hoeff}		
	Let $X=\sum_{i=1}^m X_i$ be a sum of $m$ independent random variables
	with $a_i \leq X_i \leq b_i$ for all $i$.  Then
	\begin{align}\label{eq:hoeff}
		\Pr{|X - \E{X}| \geq b } \leq \exp\left(-\frac{2b^2}{\sum_{i=1}^m
		(b_i - a_i)^2}\right).
	\end{align}
\end{theorem}

 	The following Theorem finds its origins in the work of \cite{McDiarmid1998}.
\begin{theorem}[{\cite[Theorem 6.1]{chung2006}}]\label{pro:fancyazuma}
	Let $X$ be the martingale associated with a filter $\mathcal{F}$ satisfying
	\begin{enumerate}
	\item $\Var{X_i~|~\mathcal{F}_{i-1}} \leq \sigma_i^2$, for $1\leq i \leq m$;
	\item $|X_i - X_{i-1}| \leq M$, for $1 \leq i \leq m$.
	\end{enumerate}
	Then we have 
	\[ \Pr{X - \E{X} \geq b } \leq \exp\left(-\frac{b^2}{2\left(\sum_{i=1}^m
		\sigma_i^2  +  Mb/3\right)}\right).
\] 
\end{theorem}

 \begin{theorem}[{\cite[Theorem 6.5]{chung2006}}]\label{pro:fancyazumalower}
	Let $X$ be the martingale associated with a filter $\mathcal{F}$ satisfying
	\begin{enumerate}
	\item $\Var{X_i~|~\mathcal{F}_{i-1}} \leq \sigma_i^2$, for $1\leq i \leq m$;
	\item $X_{i-1} -M -a_i \leq X_i$, for $1 \leq i \leq m$.
	\end{enumerate}
	Then we have 
	\[ \Pr{X \leq \E{X}- b } \leq \exp\left(-\frac{b^2}{2\left(\sum_{i=1}^m(
		\sigma_i^2 +a_i^2)  +  Mb/3\right)}\right).
\] 
\end{theorem}

Throughout this paper we will frequently make use of the fact that the sum of independent variables is a martingale. 
 
 \arxiv{
 In the following we state an improved version of the Berry-Esseen theorem.
  
 }
%
%

\section{Missing Proofs: Sequential Setting (\autoref{sec:seq})}
\label{sec:missing}
\subsection{Sequential Setting: One Step Potential Change }

In the following we bound the one step potential change.
\begin{lemma}\label{lem:onestep}
Fix an arbitrary potential at time $t$. Suppose the pair $i,j$ was chosen to communicate and that the coins have been flipped to determine the noise, \ie $N_i=n_i$ and $N_j=n_j$. Then,
\[
\bar{\phi}(\Valvec{t+1})-\bar{\phi}(\valvec{t})  =   - \frac{(x_i-x_j)^2}{2}+  \frac{n_i^2 + n_j^2}{4} - \frac{(n_i + n_j)^2}{4n} + (n_i + n_j)   \left( \frac{x_i+x_j}{2} - \avg{t} \right)
\]
\end{lemma}
\begin{proof}
In order to bound $\bar{\phi}(\Valvec{t+1})-\bar{\phi}(\valvec{t}) $ we will make use of \autoref{lem:relTSS} and analyze 
${\phi}(\Valvec{t+1})-{\phi}(\valvec{t}) $, which is slightly more convenient, since  we do not need to compute the change of $\avg{t+1}$.

Note that besides node $i$ and node $j$ no other nodes will change their value. However, the contribution of each agent to the potential might change. Consider for  $k\not\in\{i,j\}$
\begin{align*}
\left( \frac{x_i+x_j+n_i}{2}-x_k\right)^2 
		&=\frac{1}{2}  x_j n_i+\frac{1}{2} x_i x_j- x_kn_i-x_i x_k\\
		&+\frac{1}{2}  x_in_i+\frac{n_i^2}{4}+\frac{x_i^2}{4}-x_j
x_k+\frac{x_j^2}{4}+x_k^2 \\
& = \frac{n_i}{2}(x_i+x_j-2x_k)+\frac{1}{2} x_i x_j-x_i x_k+\frac{x_i^2}{4}-x_j
x_k+\frac{x_j^2}{4}+x_k^2+\frac{n_i^2}{4}
\end{align*}
The same holds if we substitute $n_i$ with $n_j$ and thus, for agent $k\not\in\{i,j\}$ the change in the contribution to the potential equals:
\begin{align*}
&\left( \frac{x_i+x_j+n_i}{2}-x_k\right)^2 + \left( \frac{x_i+x_j+n_j}{2}-x_k\right)^2  - \left(x_i-x_k\right)^2 -\left(x_j-x_k\right)^2
		\\&= \frac{n_i+n_j}{2}(x_i+x_j- 2 x_k)+x_i x_j-2x_i x_k+\frac{x_i^2}{2}-2x_j
x_k+\frac{x_j^2}{2}+2x_k^2+\frac{n_i^2}{4}+\frac{n_j^2}{4}- \left(x_i-x_k\right)^2 -\left(x_j-x_k\right)^2\\
\\&= x_i x_j-\frac{x_i^2}{2}-\frac{x_j^2}{2}+\frac{n_i^2+n_j^2}{4}+\frac{n_i+n_j}{2}(x_i+x_j- 2 x_k)\\
\\&= -\frac{(x_i-x_j)^2}{2}+\frac{n_i^2+n_j^2}{4}+\frac{n_i+n_j}{2}(x_i+x_j-2x_k)
\end{align*}

	 We get that if in step $t$ an interaction between $i$ and $j$ happens, and $N_i=n_i, N_j=n_j$ that the change of the potential is as follows:
	\begin{align*}
		&{\phi}(\Valvec{t+1})-{\phi}(\valvec{t})=\\
		&=2\left[ \left(\frac{x_i+x_j+n_i}{2}-\frac{x_i+x_j+n_j}{2}\right)^2-\left( x_i - x_j \right)^2\right]\\
		&\phantom{--}+
		2\sum_{k \in [n] \setminus \{ i,j\}} 
		\left( \left( \frac{x_i+x_j+n_i}{2}-x_k\right)^2 + \left( \frac{x_i+x_j+n_j}{2}-x_k \right)^2 -
		\left( x_i - x_k \right)^2 -
		\left( x_j - x_k \right)^2\right)\\
		&=
		2\left(\frac{n_i-n_j}{2} \right)^2 -
		2\left( x_i - x_j \right)^2
		+
		2\sum_{k \in [n] \setminus \{ i,j\}} \left( -\frac{(x_i-x_j)^2}{2}+\frac{n_i^2+n_j^2}{4}+\frac{n_i+n_j}{2}(x_i+x_j-2x_k) \right)\\
		&= 	-n(x_i-x_j)^2+\frac{n_i^2-2n_in_j+n_j^2}{2}+2\sum_{k \in [n] \setminus \{ i,j\}}\left(\frac{n_i^2+n_j^2}{4}+\frac{n_i+n_j}{2}(x_i+x_j-2x_k)\right)\\
		&= 	-n(x_i-x_j)^2+(n-1)\frac{n_i^2+n_j^2}{2}-n_in_j+\sum_{k \in [n] \setminus \{ i,j\}} (n_i+n_j)(x_i+x_j-2x_k)
	\end{align*}

\begin{align*}
	&= 	-n(x_i-x_j)^2+(n-1)\frac{n_i^2+n_j^2}{2}-n_in_j+ (n_i+n_j)\cdot  \left(n x_i+n x_j - 2\sum_{k \in [n]} x_k \right)  
	\\
	&= 	-n(x_i-x_j)^2+(n-1)\frac{n_i^2+n_j^2}{2}-n_in_j
	+(n_i + n_j)\cdot 2 n \cdot\left( \frac{x_j + x_j}{2} - \avg{t} \right)  \\
	&= 	-n(x_i-x_j)^2+n\frac{n_i^2+n_j^2}{2}-\frac{(n_i+n_j)^2}{2}
	+(n_i + n_j)\cdot 2 n \cdot\left( \frac{x_j + x_j}{2} - \avg{t} \right).
	\end{align*}

By \autoref{lem:relTSS} and via dividing the equation by $2n$ we obtain:

\[
\bar{\phi}(\Valvec{t+1})-\bar{\phi}(\valvec{t})  =   - \frac{(x_i-x_j)^2}{2}+  \frac{n_i^2 + n_j^2}{4} - \frac{(n_i + n_j)^2}{4n} + (n_i + n_j)   \left( \frac{x_i+x_j}{2} - \avg{t} \right). \qedhere
\]
\end{proof}
Recall, that by definition (see \autoref{sec:additional})
$N'^{(t+1)}=\left(N_1^{(t+1)}\right)^2 + \left(N_2^{(t+1)}\right)^2$ and $N^{*(t+1)}=N_1^{(t+1)} + N_2^{(t+1)}$ where $N_1^{(t+1)}$ and $N_2^{(t+1)}$ are the random variables that determine the noise of the communication in time step $t+1$.
Conditioning on all the randomness that has happened until time $t$, \ie conditioning on $\mathcal{F}_t$, we define the $\Delta^{(t+1)}$ as follows:

$\Delta^{(t+1)}  = \begin{cases}
\frac{\left(\val{i}{t}-\val{j}{t}\right)^2}{2\bar{\phi}(\valvec{t})} & \text{ for $\bar{\phi}(\valvec{t}) > 0$}\\
1/n & \text{ otherwise }
\end{cases} $, where $i$ and $j$ are the chosen agents in round $t$.

Using the above, we can prove the following two statements. 
\begin{proof}[Proof of \autoref{lem:exactdistribution}]
The first result follows with the definition of $N'^{(t+1)}, N^{*(t+1)}$ and  $\Delta^{(t+1)}$ and with \autoref{lem:onestep}.

Note that the term $ - \frac{(n_i + n_j)^2}{4n}$ is always negative.
Hence, by \autoref{lem:onestep}, we get that for fixed $i$ and $j$
\[ \bar{\phi}(\Valvec{t+1})-\bar{\phi}(\valvec{t}) \fnote{ \leq^{\text{st}} }\leq  - \Delta^{(t+1)} \cdot \bar{\phi}(\valvec{t}) +  \frac{{N'}}{4}+ N^* \left(\frac{ \val{i}{t}+\val{j}{t} }{2} - \avg{t} \right) .\]

W.l.o.g. assume $\bar{\phi}(\valvec{t}) >0$; otherwise the claim follows trivially (by definition of $\Delta^{(t+1)}$).
Taking the expectation over all choices of $i$ and $j$: 
\begin{align*}
		 \E{\Delta^{(t+1)}~|~\mathcal{F}_t;\bar{\phi}(\valvec{t}) >0}  &= 	\frac{1}{\bar{\phi}(\valvec{t}) }\frac{1}{n^2}\sum_{i,j} \frac{\left(\val{i}{t}+\val{j}{t}\right)^2}{2}  
		= \frac{1}{2 n^2}\frac{ {\phi}(\valvec{t}) }{ \bar{\phi}(\valvec{t})}  =\frac{1 }{n}  .
		 		\end{align*}

Note that in the Gaussian noise model, all $N_i$ follow the same law $\mathcal{N}(0,\sigma^2)$.
Thus, using that the sum of two Gaussian with law $\mathcal{N}(0,\sigma^2)$ are distributed
$\mathcal{N}(0,2\sigma^2)$, we obtain
 $N^* \stackrel{d}{=} N_i+N_j \sim \mathcal{N}(0,2\sigma^2)$.
Finally, the sum of two squared Gaussians each with distribution  $N_i,N_j \sim \mathcal{N}(0,\sigma^2)$ we have  
$N' \stackrel{d}{=} N_i^2 + N_j^2  \sim 2 \Gamma(1/2, 2\sigma^2) =  \Gamma(1, 2\sigma^2) $.
\fnote{we never talk about this part (Gaussian special case). We used to at some point. should we just remove it?}
\end{proof}

For an arbitrary time interval $\mathcal{I}$ define 
\[S'(\mathcal{I})= \sum_{\tau \in \mathcal{I}}  N'^{(\tau)}/4, \hspace{0.7cm} S^*(\mathcal{I})= \sum_{ \tau \in \mathcal{I}} N^{*{(\tau)}}  \left(\frac{\val{i}{\tau-1}+ \val{j}{\tau-1} }{2} - \avg{\tau} \right), \hspace{0.7cm}S^-(\mathcal{I})= \sum_{ \tau \in \mathcal{I}} \Delta^{(\tau)}~.\]
Note that $i$ and $j$ in the definition of $S^*$ also depend on $\tau$ and are the nodes (we use nodes and agents interchangeably) that are chosen in that round.
\fnote{this stuff is inbetween lemmas. Should it be?}

\begin{proof}[Proof of \autoref{lem:rewritingpotential}]
The potential
$\bar{\phi}(\Valvec{t_1})$ is maximized, if all decreases happen at the beginning and all increases happen in the last time step.\dnote{When? What? Why? Increases of what? }\fnote{ I think you are right and we should drop the crap with decreases at the beginning since Weierstrass tells us they should be split evenly.  but it should definitely be the case that increase at the end hurt as the most right?}
\fnote{should we cite \autoref{lem:exactdistribution} here}
\fnote{Reivewer 1 seems to say that this is stochastic dominance here? correct? if so udate lemma}

In order to analyze the decrease due to $S^-$ we will make use of the Weierstrass Product Inequality (\autoref{thm:weierstrass}) to derive

\begin{align*}
\bar{\phi}(\Valvec{t_1}) &\leq \prod_{\tau \in \mathcal{I}} (1-\Delta^{(\tau)}) \bar{\phi}(\Valvec{t_0})   +  S'(\mathcal{I}) + S^*(\mathcal{I}) \leq   \left(1- \frac{\sum_{\tau \in \mathcal{I}}\Delta^{(\tau)}  }{t}\right)^t \bar{\phi}(\Valvec{t_0})  +  S'(\mathcal{I}) + S^*(\mathcal{I})\\
&=   \left(1- \frac{ S^-(\mathcal{I}) }{t}\right)^t \bar{\phi}(\Valvec{t_0})  +  S'(\mathcal{I}) + S^*(\mathcal{I}). &\qedhere
\end{align*} 
\end{proof}

%
%

\subsection{Bounding $S'$, $S^*$, and $S^-$ }
In this section we bound the terms of \autoref{lem:rewritingpotential} separately. Therefore, for any $\delta \in [0,1]$, any $t_0$ and any $t$ we define the following values:
\begin{align}
m'_{t,\delta} &  = \arg\max_\ell \left\{    \Pr{\max\left\{N'^{(t_0)}, \dots , N'^{(t_0+t)}\right\} \leq \ell} \geq 1-\delta   \right\} \label{def:mprime}~,\\
m^*_{t,\delta} & = \arg\max_\ell \left\{    \Pr{\max\left\{N^{*(t_0)}, \dots , N^{*(t_0+t)}\right\} \leq \ell} \geq 1-\delta   \right\} \label{def:mstar}~,\\ 
m_{t,\delta} & = \max\{ m'_{t,\delta} ,m^*_{t,\delta} \}\label{def:m}~,\\ 
 b_{t,\delta/2}' & =\frac{t}{4}\E{N'} +\frac{2\ln(2/\delta) m'_{t,\delta/2} }{3} + \sqrt{ \frac{\ln(2/\delta)\Var{N'}t}{8}  }~,  \label{def:bprime} \\
 z & =   \bar{\phi}(\valvec{t_0}) +\frac{2\ln(2t/\delta) t    \E{\left( N^{*}\right)^2 }     }{n} + \left(\frac{2}{3}\ln(2t/\delta)m_{t,\delta/4} \right)^2 + b'_{t,\delta/4} + 1,  \label{def:z}\\
b_{t,\delta/4}^* & =  \left(\frac{2\ln(2t/\delta)m^*_{t,\delta/4}}{3}+ \sqrt{ \frac{2\ln(2t/\delta) t    \E{ \left( N^{*}\right)^2 }     }{n} }  \right)\cdot \sqrt{z}~. \label{def:bstar}
\end{align}
We want to emphasize that these values are not random variables and their values are not related to the actual outcome of the randomness during a run of the protocol. 
In words $m'_{t,\delta}$ ($m^*_{t,\delta}$, respectively)  denotes the maximum value that is reached w.p. at most $\delta$ by $N'$ and $N^*$ during the interval $[t_0,t_0+t]$, respectively. Note, that the value of $m'_{t,\delta}$ and $m^*_{t,\delta}$ is independent from the choice of $t_0$ as the noise at different time steps is independent.
We will assume $m_{t,\delta}\geq 1$ throughout the proofs; we will only consider $t$ that are a function of $n$ and hence we restrict ourselves to noise functions that grow with $n$.

From now on and throughout the proof assume that $\delta>0$ is fixed and we continue with upper bounding $S'$.

\begin{lemma}\label{lem:s}
Fix  arbitrary $t_0, t_1$, 
consider the interval $\mathcal{I} = (t_0,t_1]$ and let $b'=b_{t,\delta/2}'$ be the value as defined in \autoref{def:bprime} where $t=t_1-t_0$.
Then, w.p. at least $ 1-\delta$ we have
\[ S'(\mathcal{I}) \leq  b'~.\]
\end{lemma}
\begin{proof}
By definition $m'_{t,\delta/2}$ w.p. at least $1-\delta/2$ for every $t'\in\mathcal{I}$ we have $N'^{(t')} \leq  m'_{t,\delta/2}$. Assume that this property holds throughout interval $\mathcal{I}$. 

Define $X_{t'}=S'((t_0,t'])-\E{S'((t_0,t'])}$ and note that $\big(X_{t'}\big)_{t_0<t'\leq t_1}$ is a martingale. 
We obtain that $|X_{t'}-X_{t'-1}| \leq N'^{(t')}/4\leq M$ 
and $\Var{X_{t'}~|~\mathcal{F}_{t'-1}}=\Var{N'/4}= \Var{N'}/16$.
Let $b''=b'-\frac{t}{4}\E{N'}$.
We have that
\begin{align*}
  (b'')^2 &= \left(\frac{2\ln(2/\delta) m'_{t,\delta/2} }{3} + \sqrt{ \frac{\ln(2/\delta)\Var{N'}t}{8}  } \right)b'' \geq     \frac{2\ln(2/\delta) m'_{t,\delta/2} }{3} b'' +   \frac{\ln(2/\delta)\Var{N'}t}{8}   
  \end{align*}
 and apply \autoref{pro:fancyazuma} to the martingale which yields



\begin{align*} 
 \Pr{S'(\mathcal{I}) \geq b' }  &=
\Pr{S'(\mathcal{I}) - \E{S'(\mathcal{I})} \geq b''}  
\leq 
\exp\left(-\frac{\big( b''\big)^2}{2\left( \frac{\Var{N'} t}{16 } + m'_{t,\delta/2} b''/3\right)}\right) \\
&\leq
\exp\left(-\frac{   \frac{2\ln(2/\delta) m'_{t,\delta/2}b'' }{3} +  \frac{\ln(2/\delta)\Var{N'}t}{8}      }{2\left( \frac{\Var{N'} t}{16 } + m'_{t,\delta/2} b''/3\right)}\right) = \exp(-\ln(2/\delta))=\frac{\delta}{2}~.
\end{align*} 
Taking a union bound with the case that $N'^{(t')}$ is not smaller than $M$ for some $t'\in \mathcal{I}$, yields the claim.
\end{proof}

  In the following we bound the first, second moment of $Z^{(t+1)}= \frac{\Val{i}{t}+ \Val{j}{t} }{2} - \avg{t}$ and its maximum possible value; we use this result in the proof of \autoref{lem:sprime}. 
  \begin{fact}\label{fact:ingvariable}
  Fix $\mathcal{F}_t$. In particular, this fixes the vector of values $\valvec{t}, \avg{t}$ and $\bar{\phi}(\valvec{t})$.
  Define the following random variable $Z^{(t+1)}= \frac{\Val{i}{t}+ \Val{j}{t} }{2} - \avg{t}$, where $i$ and $j$ are chosen uniformly at random. We have:
  \begin{enumerate}
  \item $ \E{Z^{(t+1)}~|~\mathcal{F}_t} =0$
  \item $ \E{\left(Z^{(t+1)}\right)^2~\middle|~\mathcal{F}_t} \leq \frac{ \bar{\phi}(\valvec{t}) }{ n} $
    \item $ Z^{(t+1)}\leq \sqrt{\bar{\phi}(\valvec{t}) }$ \label{fact:upperBound}
  \end{enumerate}
  \end{fact}
  \begin{proof}
  Recall that we allow $i=j$.
Using this, we derive 
  \begin{align*}
\E{Z^{(t+1)}~|~\mathcal{F}_t} &=  2\E{ \frac{\Val{i}{t} }{2} - \frac{\avg{t}}{2} ~\middle |~\mathcal{F}_t} =  \E{ \Val{i}{t} ~\middle |~\mathcal{F}_t} -  \avg{t} = 0.
\end{align*} 
Moreover, using that $(a+b)^2 \leq 2a^2 + 2b^2$, we get
  \begin{align*}
    \E{ \left(Z^{(t+1)}\right)^2 ~|~\mathcal{F}_t} &\leq  4\E{ \left(\frac{\Val{i}{t} }{2} - \frac{\avg{t}}{2}\right)^2 ~\middle |~\mathcal{F}_t} =
   \E{ \left(\Val{i}{t} - \avg{t}\right)^2~\middle |~\mathcal{F}_t} \\
     &=\sum_{i} \frac{1}{n} (\val{i}{t}-\avg{t})^2= \frac{ \bar{\phi}(\valvec{t}) }{n} .
\end{align*} 
The third claim is true because for any $i$ we have 
\[
\Val{i}{t} - \avg{t}  \leq \sqrt{(\Val{i}{t} - \avg{t})^2} \leq \sqrt{\sum_{j}(\Val{j}{t} - \avg{t})^2} =\sqrt{\bar{\phi}(\valvec{t}) }. \qedhere
\]

  \end{proof}
We continue with upper bounding $S^*$.
  
\begin{lemma}\label{lem:sprime}
Fix  arbitrary $t_0, t_1$, consider the interval $\mathcal{I}=(t_0,t_1]$ and let $b^*=b_{t,\delta}^*$ be the value as defined in \autoref{def:bstar} with $t=t_1-t_0$.
Then, w.p. at least
$ 1-\delta$,  we have \[ S^*(\mathcal{I}) \leq  b^*~. \]
\end{lemma}
\begin{proof}
By the definition of $m^*_{t,\delta/4}$ (see \autoref{def:mprime}) we have $N^{*{(\tau)}} \leq m^*_{t,\delta/4}$ w.p. $1-\delta/4$ throughout $\mathcal{I}$. Further by \autoref{lem:s} with probability $1-\delta/4$ we have $S'((t_0,t'])\leq S'((t_0,t_1])\leq b'_{t,\delta/4}$~ for all $t'\in (t_0,t_1]$. We assume that both properties hold (the case that they do not hold is submerged in a union bound (that leads to a probability $\leq \delta$) with all other undesirable cases).

Recall the definition $Z^{(\tau)}= \frac{\val{i}{\tau-1}+ \val{j}{ \tau-1 } }{2} - \avg{\tau-1} $ as in \autoref{fact:ingvariable} and note that for $t' \in [t_0, t_1]$, we have $S^*((t_0,t'])= \sum_{ \tau \in (t_0,t']} N^{*{(\tau)}} Z^{(\tau)}$. 
For each such $t'$ the sequence $\big(S^*((t_0,\tau])\big)_{t_0\leq \tau\leq t'}$  is a martingale and the goal is to apply \autoref{pro:fancyazuma} to it. 

\smallskip

We assume a process $P^*$ in which $\bar{\phi(\Valvec{t'})} \leq z$ for all $t' \in [t_0,t_1)$, where $z$ is defined as in \eqref{def:z}.
In this process we will bound the size of $S^*$. Using this bound, we show that the original process $P$ and $P^*$ never diverge (with large probability) and hence the bound on $S^*$ we obtained in $P^*$ carries over to $P$.

Consider $P^*$. Using that the potential is at most $z$ and  \autoref{fact:ingvariable}, we get a bound  of $\frac{z}{n}$ on the second moment of $Z^{(\tau)}$ (conditioned on $\mathcal{F}_{\tau-1}$). 
Using this bound and since $Z^{(\tau)}$ and $N^{*{(\tau)}}$ are independent and as $\E{Z^{(\tau)}~|~\mathcal{F}_{\tau-1}}$ and $\E{N^{*{(\tau)}}~|~\mathcal{F}_{\tau-1}}=\E{N^{*{(\tau)}}}$ equal $0$ we obtain that 
\dnote{Is the considered interval $(t_0,t_1]$ or $[t_0,t_1]$? The latter has length $t_1 - t_0 +1$ so probably the former... But it is mixed up so many times that I got confused. }
\fnote{pretty sure it should be the open interval. I tried to repalce as much as I could fine}
\begin{align*}\Var{S^*((t_0,\tau])~|~\mathcal{F}_{\tau-1}} & =  \Var{N^{*{(\tau)}} Z^{(\tau)}~|~\mathcal{F}_{\tau-1}} \leq \E{ \left( N^{*}\right)^2 ~|~\mathcal{F}_{\tau-1}} \E{ \left(Z^{(\tau)} \right)^2~|~\mathcal{F}_{\tau-1}} \\
 &\leq \E{ \left( N^{*{(\tau)}}\right)^2}\frac{z}{n}~. 
\end{align*}

Define $M= m^*_{t,\delta/4}\sqrt{z} $. As the potential is never above $z$ we obtain $Z^{(\tau)}\leq \sqrt{z}$ for all $\tau\in (t_0,t]$ due to \autoref{fact:ingvariable}, \ref{fact:upperBound}. Due to the definition of $m^*_{t,\delta/4}$ this implies that 
\begin{align}
S^*((t_0,\tau])-S^*((t_0,\tau-1])=N^{*{(\tau)}}Z^{(\tau)} \leq M \label{eqn:cond2}
\end{align} throughout the interval $(t_0,t']$.

The bound on the variance and \autoref{eqn:cond2} are sufficient to apply \autoref{pro:fancyazuma}
and we obtain 

 \begin{align*}
\Pr{S^*(t_0,t']) \geq b^* }  & = 
\Pr{S^*((t_0,t']) - \E{S^*([t_0,t'])} \geq b^* }   \\
&\leq \exp\left(-\frac{(b^*)^2}{2\left( (t'-t_0)\E{ \left( N^{*{(\tau)}}\right)^2 }\frac{z}{n}+ M  b^*/3\right)}\right) \\
&\leq \exp\left(-\frac{(b^*)^2}{2\left( t \E{ \left( N^{*{(\tau)}}\right)^2 }\frac{z}{n}+  M b^*/3\right)}\right)~,
\end{align*}
where we used that $\E{S^*((t_0,t'])} =0$.  
Let $A=\sqrt{ \frac{2\ln(2t/\delta) t \E{ \left( N^{*}\right)^2 }  }{n} }$ and $B=\frac{2\ln(2t/\delta)m^*_{t,\delta/4}}{3}$. Then we have
$(b^*)^2 =b^*\cdot (A+B)\cdot \sqrt{z} =   b^* A\sqrt{z} +b^* B\sqrt{z}\geq A^2z+Bb^*\sqrt{z}~, $ which yields

 \begin{align*}
\Pr{S^*(t_0,t']) \geq b^* }  \leq   \exp\left(-\frac{      \frac{2\ln(2t/\delta) t    \E{ \left( N^{*}\right)^2 }     }{n}   + \frac{2\ln(2t/\delta)m^*_{t,\delta/4}}{3}\sqrt{z} b^* }{2\left( t \E{ \left( N^{*{(\tau)}}\right)^2 }\frac{z}{n}+  \frac{m^*_{t,\delta/4}\sqrt{z} \cdot b^*}{3}\right)}\right)   =
 \exp(-\ln(2t/\delta)  )= \frac{\delta}{2t}.
\end{align*}
To prove the equivalence  between the processes, we need to show that the potential at step $t'$ does not exceed $z$; in this proof we use $S^*((t_0,t'])\leq b^*$ and $S'((t_0,t'])\leq b'_{t,\delta/4}$~. The first statement holds with probability $1-\delta/(2t)$ as we just showed and we assumed $S'((t_0,t'])\leq b'_{t,\delta/4}$ to hold throughout the whole proof at the very beginning. Thus we obtain

\begin{align*}
\bar{\phi}(\valvec{t'}) &\leq \bar{\phi}(\valvec{t_0}) +  S'((t_0,t']) + S^*((t_0,t']) \leq \bar{\phi}(\valvec{t_0})  + b'_{t,\delta/4}+ b^* \leq z,
\end{align*} 

where the last inequality follows as $x\sqrt{z'} \leq z'$ for all $z' \geq x^2$.

The last induction step shows $\Pr{S^*(t_0,t']) \geq b^* }\leq \delta/(2t)$. This combined with a union bound about the error probabilities from the two assumptions at the start of the lemma ($\delta/4$ each)  yield that the result holds with probability $1-\delta$.
\end{proof}

Now, we lower bound $S^-$ which is essential to obtain progress through \autoref{lem:rewritingpotential}.

\begin{proof}[Proof of \autoref{lem:bounds}, \ref{lem:sminus}]Let $\tau\in \mathcal{I}$.
Note that $\Delta^{(\tau)} \in [0,1]$. By \autoref{lem:exactdistribution}, we have $\E{\Delta^{(\tau)}}=\frac{1}{n}$ and $\Var{\Delta^{(\tau)}} \leq \E{(\Delta^{(\tau)})^2}\leq \E{\Delta^{(\tau)}} \leq 1^2\cdot \frac{1}{n} = \frac{1}{n} $, where  Why $\E{(\Delta^{(\tau)})^2}\leq \E{\Delta^{(\tau)}}$ follows since $\Delta^{(\tau)} \in [0,1]$ implies that each element of the sum is smaller.
By \autoref{pro:fancyazumalower} with $M=1, a_i= 0$ for all $i$, with $b=\frac{\gamma t}{n}$, we get 
\begin{align*}
\Pr{S^-(\mathcal{I})  \leq t/n (1-\gamma)} & =
\Pr{S^-(\mathcal{I})  \leq t/n - \gamma t/n } = 
\Pr{S^-(\mathcal{I})  \leq \E{S^-(\mathcal{I})} - b } \\
&\leq \exp\left(-\frac{b^2}{2\left( \frac{t}{n} + b/3 \right)}\right)
= \exp\left(-\frac{\frac{\gamma ^2t^2}{n^2}}{2\left( \frac{t}{n} + \frac{\gamma t}{3n} \right)}\right)
= \exp\left(-\frac{\gamma ^2t}{2n\left( 1 + \frac{\gamma }{3} \right)}\right)\\
&\leq  \exp\left(-\frac{3\gamma ^2t}{8n}\right). & \qedhere
\end{align*}

%

\end{proof}

\begin{proof}[Proof of \autoref{lem:bounds} \ref{lem:combinedugliness}]
In order to derive the result, we the bound on $S'$ (\autoref{lem:s}) and on  $S^*$ (\autoref{lem:sprime}).
%
%
%
$\autoref{lem:sprime}$ applied with $\delta/4$  and for $t\geq n$ and $n$ large enough with probability $1-\delta/4$   yields:
\arxiv{Frederik has to check this. I couldn't verify the version. I added one intermediate line that made it much clearer for me but I am not sure that it is correct.}
Roughly upper bounding the terms yields
\begin{align*} S^*(\mathcal{I}) \leq b_{t,\delta/4}^*  &=  \left(\frac{2\ln(4t/\delta)m^*_{t,\delta/4}}{3}+ \sqrt{ \frac{2\ln(4t/\delta) t    \E{ \left( N^{*}\right)^2 }     }{n} }  \right)\cdot\sqrt{z}\\
&\leq 4\sqrt{\frac{t}{n}} \left(\ln(4t/\delta) m^*_{t,\delta/4}\right) \left(1+ \sqrt{\E{(N^*)^2}} \right)  \sqrt{z}.
\end{align*}
Furthermore, we have, by the definition of $z$ (\autoref{def:z}), 

\begin{align*}
\sqrt{z} &\leq \ln(4t/\delta) m^*_{t,\delta/4} \sqrt{ \bar{\phi}(\valvec{t_0} + \frac{2t}{n} \E{(N^*)^2}  + \frac{4}{9} \ln(4t/\delta) m^*_{t,\delta/4}        +  \frac{t}{4}\E{N'}   +\frac23  + \sqrt{t\Var{N'}/8}  } \\
&\leq  \ln(4t/\delta) m^*_{t,\delta/4}  \sqrt{  \bar{\phi}(\valvec{t_0}) + 9 {t} \E{N'}    +2 },
\end{align*}
where we used that $\sqrt{\Var{N'}/8} \leq \sqrt{n} \leq \sqrt{t}$.
Note that  $\sqrt{\E{(N^*)^2} }=\sqrt{ \E{N'}} \leq 1 + \E{N'} $~.
Thus, putting everything together yields

\begin{align}\label{label:lol}
 S^*(\mathcal{I}) \leq    4\sqrt{\frac{t}{n}} \left(\ln(4t/\delta) m^*_{t,\delta/4}\right)^2 \left(2+ \E{N'} \right)   \sqrt{  \bar{\phi}(\valvec{t_0}) + 9 {t} \E{N'}    +2 }
\end{align}

Applying \autoref{lem:s} with $\delta/4$  yields that with probability $1-\delta/4$  we have
\begin{align*}\label{eq:sprimenicer}
 S'(\mathcal{I}) \leq b'_{t,\delta/4} =  \frac{t}{4}\E{N'} +\frac{2\ln(4/\delta) m'_{t,\delta/4} }{3} + \sqrt{ \frac{\ln(4/\delta)\Var{N'}t}{8}  } \leq \frac{t}{4}\E{N'}  +  S^*(\mathcal{I}). 
\end{align*}

Combining the bound on $S'(\mathcal{I})$ and $S^*(\mathcal{I})$ yields the claim with probability $1-\delta$.
\fnote{finish here}
\end{proof}

\subsection{Proof of \autoref{pro:main} and \autoref{thm:runningavg}}
\label{sec:mainPropositionProof}
\ynote{TODO}
In the section we prove \autoref{pro:main}. A proof sketch can be found in \autoref{sec:seq}. Using \autoref{pro:main} and \autoref{lem:rewritingpotential} our main theorem \autoref{thm:runningavg} follows almost immediately. 
\begin{proof}[Proof of \autoref{pro:main}]

We distinguish between regimes based on the value of $\bar{\phi}=\bar{\phi}(\Valvec{t})$ for thresholds that we define later.
We have two regimes:  regime $(2)$ starts at time $t_0$ and ends when the potential is below a threshold $b_2(1)$ which marks the start of regime $(1)$; note that in order to simplify indices in the calculations, we define regimes and phases (within the regimes) backwards, starting with large numbers and then reduce.
We divide each regime into phases. The phases of regime $(1)$ are such that  the $i$'th phase starts when $\bar{\phi} \in [b_2(i),b_2(i+1)] $. Phases in regime $(1)$ are also counted backwards starting from phase
$i_{max}=\min\{ \{i : \bar{\phi}(\valvec{t_0} ) \geq b_2(i)\} \cup \{0\} \}$ until we reach phase $0$. 
We use $\tau_i^\iota$, $\iota \in \{(1), (2) \}$ to denote the start of the $i$'th phase of regime $\iota$. Where the first phase is $\tau^{(2)}_{i_{max}}.$ 


%
Let $\gamma^*=1-e$, $c^* =  \frac{8}{3\gamma^*} \leq 0.91$. We define the boundaries $b_1$ and $b_2$ (which will be used to guide the potential decrease) as follows. Let $c^{**}$ be a large enough constant and $\varepsilon = 1/20$.

\[b_2(i) = \left( 100 \left(\frac{4}{3}\right)^{15} \left(\ln(n/\delta) \right)^{15} \left(\frac{n}{\delta}\right)^{13\varepsilon}  n\left(2+ \E{N'} \right)^6  \right)^{\left(\frac{4}{3}\right)^{i+1} },\]
 \[b_1 =   c^*\ln(c^{**}1/\delta) n \frac{\E{N'}}{4} +  \left(1+ \E{N'} \right)\left(\ln(1/\delta) \right)^{9} n^{9/10}.\]


\paragraph{Regime 2}
Consider phase $i$, that is, the potential $\bar{\phi}=\bar{\phi}(\valvec{\tau^2_i})$ at the start of the phase is in the interval $[b_2(i),b_2(i+1)]  $. 
Let $ t_i = 100  n \ln\left( \frac{b_2(i+1)}{ \delta }\right)$.

In the following we bound the increase due to $S^*+S^-$ after $t_i$ steps.
By \autoref{lem:bounds}, \ref{lem:combinedugliness}, after $ t_i $ time steps we have
that each of the following bound holds w.p. at least $1- \frac{\delta}{ b_2(i+1) }$.

\begin{align}\label{eq:boundthis}
S^*+S^- \leq   \frac{t_i}{4}\E{N'} +\ 5\sqrt{\frac{t_i}{n}} \left(\ln(4t_i/\delta) m_{t,\delta/4}\right)^2 \left(2+ \E{N'} \right)   \sqrt{  \bar{\phi} + 9 {t_i} \E{N'}    +2 }~.
\end{align}
In the following we bound the terms of \eqref{eq:boundthis}. First we obtain
\begin{align}\label{eq:messy1}
\frac{t_i}{4}\E{N'} &= 25 n \ln(b(i+1)/\delta)\E{N'} \notag\\
&=  25 n  \E{N'}  \ln(b^{4/3}_2(i)/\delta)\notag\\
&=  25  n \E{N'}  4/3 \ln(b_2(i)/\delta) \notag\\
&\stackrel{*_1}{\leq} 25 n\E{N'} \frac{ b^\varepsilon_2(i)}{\delta^{\varepsilon}}= \frac14 \cdot 100 \frac{n\E{N'}}{\delta^{\varepsilon}} b^\varepsilon_2(i) \notag \\
&\stackrel{*_2}{\leq} \frac14  b^{1-\varepsilon}_2(i) b^\varepsilon_2(i) =  \frac{b_2(i)}{4} ,
\end{align}
Here,  $*_1$ follows because $\ln(x)\leq x^{\varepsilon}/400$ for large enough $x$ and 
$*_2$ follows because $100\cdot\frac{n\E{N'}}{\delta^{\varepsilon}}\leq b_2^{1-\eps}(i)$~.

We continue by bounding the terms in front of the square-root of \eqref{eq:boundthis}.
To do so we consider the factors separately. Due to $\sqrt{x}\leq x$ and $\ln x\cdot y\leq \ln x\cdot \ln y$ for large enough $x,y>0$ we have
 \[\sqrt{t_i/n} = 10 \sqrt{  \ln\left( \frac{b_2(i+1)}{ \delta }\right) } \leq  10 \ln(b_2(i+1)) \ln(1/\delta)\leq  10 \ln(b_2(i+1)) \ln(n/\delta).\]
  Due to the smoothness of the noise distribution, we have
\[  m_{t_i,\delta/4}^2 \leq  \left(4\frac{t_i}{\delta}\right)^{2\varepsilon}  
= \left(\frac{n}{\delta}\right)^{2\varepsilon} \cdot 400^{2\varepsilon}       \ln^{2\varepsilon}(b_2(i+1)).
 \]
 Moreover, again using $\ln x\cdot y\leq \ln x\cdot \ln y$ for $x,y\geq 1$, we can bound 
\begin{align*} 
 \ln^{2}( 4 t_ i/\delta ) &\leq 
 \ln^{2}( 4\cdot 100\cdot n\cdot b_2(i+1)/\delta) \leq  \ln^2(400) \cdot  \ln^{2}(  n /\delta) \cdot \ln^2(b_2(i+1))
\end{align*}
Putting everything together, using that $b_2(i+1)=b^{4/3}_2(i)$. 
\begin{align}\label{eq:messy2}
 5\sqrt{\frac{t_i}{n}} &\left(\ln(4t_i/\delta) m_{t_i,\delta/4}\right)^2 \left(2+ \E{N'} \right)\notag\\
& \leq  5\cdot \left( 10 \ln(b_2(i+1)) \ln(n/\delta)\right)\cdot \left(1+ \E{N'} \right) \notag\\
& \hspace{1cm}\cdot \Bigl[\ln^2 (400) \cdot \ln^{2}(  n /\delta) \cdot \ln^2(b_2(i+1)) \cdot \Bigr]\Bigl[\left(\frac{n}{\delta}\right)^{2\varepsilon} \cdot 400^{2\varepsilon}       \ln^{2\varepsilon}(b_2(i+1))\Bigr]   \notag \\
 &\leq   50 \cdot\ln^2(400)\cdot 400^{2\varepsilon} \cdot   \left(\frac{n}{\delta}\right)^{2\varepsilon} \ln^3(n/\delta) \cdot \ln^{3+2\varepsilon}(b_2(i+1))\cdot \left(2+ \E{N'} \right) \notag\\
 &\leq     3278  b_2(i)^{1/5}     \cdot \ln^{3+2\varepsilon}(b_2(i+1))\notag \\
  &\leq     3278 b_2(i)^{1/5}     \cdot (4/3)^5 \ln^{5}(b_2(i)) \notag\\
  &\stackrel{*_1}{\leq}      3278 b_2(i)^{1/5}     \cdot (4/3)^5 \frac{(b_2(i))^{2/15}}{10^6}\notag \\
  &\leq  \frac{b_2(i)^{1/3}}{18},
\end{align}
where $*_1$ follows because $\ln(x)\leq x^{2/15}/10^6$ for large enough $x$.



Plugging \eqref{eq:messy1} and \eqref{eq:messy2}  into \eqref{eq:boundthis} yields
\begin{align}\label{eq:boundthis2}
S^*+S^- & \leq  \frac{b_2(i)}{4}
+ \frac{b_2(i)^{1/3}}{18}\sqrt{  \bar{\phi} + 9 {t_i} \E{N'}    +2 } \notag \\
  & \leq    \frac{b_2(i)}{4}
+ \frac{b_2(i)^{1/3}}{18} \sqrt{  b^{4/3}_2(i) + 9  \frac{b_2(i)}{4}    +2 } \leq \frac{b_2(i)}{2}.
\end{align}

Finally, by \autoref{lem:bounds}, \ref{lem:sminus} \fnote{why?}, $S^- \geq (1-\gamma^*)\frac{t_i}{n} \geq  \ln(b_2(i+1)) $ 
and since  $\bar{\phi}(\Valvec{\tau^{(2)}_i }) \leq b_2(i+1)$   we have \fnote{give proba. Probabbly, $\delta/b_2(i+1)$ good enough here  }
 \begin{align} \label{eqn:progress}
  \left(1- \frac{ S^- }{t_i}\right)^{t_i} \bar{\phi}(\Valvec{\tau^{(2)}_i }) \leq \exp(-S^-)\bar{\phi}(\Valvec{\tau^{(2)}_i }) &\leq  \exp(-\ln(b_2(i+1))) b_2(i+1) = 1 .
   \end{align}
By  \autoref{lem:rewritingpotential}, \autoref{eq:boundthis2} and \autoref{eqn:progress}  we obtain
\begin{align*}
\bar{\phi}(\Valvec{   \tau^{(2)}_i + t_i  }  ) &\leq   \left(1- \frac{ S^- }{t_i}\right)^{t_i} \bar{\phi}(\Valvec{\tau^{(2)}_i })  + S' + S^* \leq 1+\frac{b_2(i)}{2}\leq b_2(i)~.
\end{align*}  

The number of time steps $\sum_i t_i$ of all phases in regime $(2)$ form a geometric series and is dominated by the length of the first phase, that is, regime $(1)$ takes at most 
\[\sum_i t_i=O\left( n \ln\left( \frac{\bar{\phi(\valvec{t_0)}}}{  \delta}\right) \right)= O\left( n \ln\left( \frac{\bar{\phi(\valvec{t_0)}}}{ \E{N'} n \delta}\right) \right)\]
 time steps. 
The probability of success of each phase also forms a geometric series and by a union bound over all phases, the probability of failure is 
 $\sum_i \frac{3\delta}{ b_2(i+1) } \leq \frac{\delta}{3}$. 

%
 
%
%
%
\paragraph{Regime 1}

  $b_1 \leq \bar{\phi} <b_2(1)$. Here we define the phases informally to avoid an overload of notation involving the tower-function. Instead of reducing $\phi$ to $\phi^{3/4}$ as in a phase in the regime above, here, in a single phase, we reduce the potentital from $\phi$ to $f\ln(\phi)$ and then from $f\ln(\phi)$ to $f\ln(f\ln(\phi))$ etc., where $f$ is such that
this recursion forms a geometric series. We stop once the potential is smaller than $b_1$. 

From now on fix a phase in regime $(1)$ and assume that the potential is of value $\bar{\phi}$ at the start of the phase. Let $t_{\bar{\phi}}=c^* n \ln\left(c^{**} \frac{\bar{\phi}  }{ \E{N'} n \delta} \right)$, where $c^{**}$ is some large enough constant. \arxiv{define regime properly here if time permits}
%
If $\bar{\phi}\gg b_1$ the length of a phase is $t_{\bar{\phi}}$ and once $\bar{\phi}$ is close to $b_1$ the length of a phase is $c^*n \ln(c^{**}/\delta)$, more formally the length of a phase is $t'= \max\{t_{\bar{\phi}},  c^*n \ln(c^{**}/\delta)   \}$.
By \autoref{lem:bounds}, \ref{lem:combinedugliness}, after $ t'$ time steps we have,
\begin{align}\label{eq:messy3}
S^*+S^- &\leq   \frac{t'}{4}\E{N'} +\ 5\sqrt{\frac{t'}{n}} \left(\ln(4t'/\delta) m_{t',\delta/4}\right)^2 \left(2+ \E{N'} \right) \cdot \sqrt{  \bar{\phi} + 9 {t'} \E{N'}    +2 } 
\end{align}
In the following we bound the second term of \eqref{eq:messy3}.
Since we are in regime $(2)$ we have $\bar{\phi} \leq b_2(1)$ and we can deduce that 
$t' = O\left(
 n \ln\left( \frac{b_2(1) }{  \delta} \right)\right) \leq 
b_2(1)$ and $9 {t'} \E{N'}    +2 \leq b_2(1)$. We obtain

\begin{align}
5&\sqrt{\frac{t'}{n}} \left(\ln(4t'/\delta) m_{t',\delta/4}\right)^2 \left(2+ \E{N'} \right) \cdot \sqrt{  \bar{\phi}+ 9 {t'} \E{N'}    +2 } \notag\\
 & \hspace{1cm} \leq 5 \sqrt{\frac{b_2(1)}{n}} \ln(4b_2(1)/\delta) \left(\frac{4n}{\delta}\right)^{2\varepsilon} \left(2+ \E{N'} \right) 2 \sqrt{2b_2(1)}\notag\\
  & \hspace{1cm} \leq 10 \cdot \sqrt{2} \cdot 4^{2\varepsilon}  \cdot \ln(4b_2(1)/\delta) \left(\frac{n}{\delta}\right)^{2\varepsilon} \left(2+ \E{N'} \right)\cdot {\frac{b_2(1)}{\sqrt{n}}} \label{eqn:whatwebound}
\end{align}
To upper bound this term by $b_1$ we observe that the polynomial appearance of $n$ is $n^{9/10}$ in $b_1$ and only $n^{5/6}$ in the term $\frac{b_2(1)}{\sqrt{n}}$. All other terms do not ruin the claim and we can bound (\ref{eqn:whatwebound}) by $b_1$~.
Thus,
\begin{align*}
S^*+S^- 
&\leq \frac{t'}{4} n  \E{N'} + b_1.
\end{align*}
By  \autoref{lem:rewritingpotential} and using the lower bound on $S^-$ from \autoref{lem:bounds}, \ref{lem:sminus} we obtain
\begin{align*}
\bar{\phi}(\Valvec{\tau^{(1)}_{ \bar{\phi}} + t'}) &\leq   \left(1- \frac{ S^- }{t'}\right)^{t'} \bar{\phi}(\Valvec{ \tau^{(1)}_{ \bar{\phi}}})  + S' + S^* \\
&\leq  \bar{\phi} \exp\left(- \ln\left(8 \frac{\bar{\phi}}{  n \E{N'}  } \right) \right)    + S' + S^*    \\
& = \frac{n\E{N'}}{8}    + S' + S^* 
\intertext{
Now, if $\bar{\phi} < n \E{N'}$ we get directly that $\bar{\phi}(\Valvec{\tau^{(1)}_{ \bar{\phi}} + t'})\leq 2 b_1$ and otherwise we obtain}
\bar{\phi}(\Valvec{\tau^{(1)}_{ \bar{\phi}} + t'}) &\leq \frac{n\E{N'}}{8} +  \frac{c^*}{4} n \ln\left(c^{**}  \frac{\bar{\phi}  }{ \E{N'} n \delta}\right) \E{N'} + b_1\\
&\leq  n\E{N'}\frac{ \ln(\bar{\phi}) }{2} + b_1.
\end{align*}  
Now there are two cases. 
If $b_1 \leq n \E{N'}\frac{ \ln(\bar{\phi}) }{2}$, then $\bar{\phi}(\Valvec{\tau^{(1)}_{ \bar{\phi}} + t'})  \leq  n \E{N'} \ln(\bar{\phi}) $ and we continue with the next phase.
On the other hand, if $b_1 > n\E{N'}\frac{ \ln(\bar{\phi}) }{2}$, then we have $\bar{\phi}(\Valvec{\tau^{(1)}_{ \bar{\phi}} + t'}) \leq 2b_1$ and we are done.

%

We now calculate the success probability as well as total length of regime $2$. In the last run we set the error parameter $\delta$ to $\delta/20$. In the $i$'th run before the last run the error parameter is set to $\frac{1}{2^i} \delta/20$. Clearly, the total error sums up to at most $\delta/20$.
Thus with probability of at least $1-\delta/3$ the potential decreases to $2b_1$.

 To analyze the runtime of regime $(1)$ first consider the case that regime $(2)$ is executed before regime $(1)$ because the initial potential $\bar{\phi(\valvec{t_0)}}$ was larger than $b_2(1)$. Then we have $\log^* b_2(1)$ phases and the longest phase has length $O\left(n \ln\left(b_2(1)/\delta\right) \right)$~. Thus one can immediately bound the runtime of regime $(1)$ as $O\left(\log^* b_2(1)\cdot n \ln\left(b_2(1))\delta\right) \right)$~. As the length of the phases---ignoring the factor of $n$--- is decreasing more than  geometrically 
    a tighter analysis shows that the runtime of all phases can be bounded by $O\left(n \ln\left(b_2(1)/\delta\right) \right)$. 
In the other case that regime $(2)$ is not executed before regime $(1)$ the initial potential $\bar{\phi(\valvec{t_0)}}$ is smaller than $b_2(1)$ and we can replace all occurrences of $b_2(1)$ in the runtime analysis with $\bar{\phi(\valvec{t_0)}}$.

\paragraph{Combining Regimes and Phases}

Taking a union bound over all errors in all phases in both regimes gives an error probability of at most $\delta$. 
Note that regime 1 takes at most $O\left( n \ln\left( \frac{\bar{\phi(\valvec{t_0)}}}{ \E{N'} n \delta}\right) \right)$ rounds.
If regime 2 is necessary than $\bar{\phi(\valvec{t_0)}} \geq b_2(1) \geq (n\E{N'})^{1.1}$ and hence
we can bound the number of rounds in regime 1 by \[O\left( n \ln\left( \frac{\bar{\phi(\valvec{t_0)}}}{  \delta}\right) \right)= O\left( n \ln\left( \frac{\bar{\phi(\valvec{t_0)}}}{ \E{N'} n \delta}\right) \right)~.\]

Summing over both regime gives yields the claim.
\end{proof}

We are ready to prove the first main theorem.

\begin{proof}[Proof of \autoref{thm:runningavg}]

\textit{Proof of (i).} First observe that if $t_0=0$ and $t$ were to coincide with time $t^*$ as in \autoref{pro:main} then the proposition immediately yields the result.
 Otherwise we apply \autoref{pro:main} with an initial potential $\tilde{\phi}$ that is larger than $\bar{\phi}(\Valvec{t_0})$.  $\tilde{\phi}$ is chosen such that $t^*$ in \autoref{pro:main} equals the $t$ in \autoref{thm:runningavg}.
Choosing a larger potential than the actual analysis does not harm the correctness of \autoref{pro:main} as the proof does not consider an exact potential but always just upper bound on the potential.

\textit{Proof of (iii).} The lower bound on the expected size ($\E{\bar{\phi}(\Valvec{t})} = \omega(\sigma^2 n)$) follows from the following argument.
 By  \autoref{lem:onestep}, summing over all pairs of nodes, we get 
\[\E{\bar{\phi}(\Valvec{t})~|~\mathcal{F}_t; \bar{\phi}(\Valvec{t-1}) = \bar{\phi}(\valvec{t})} \geq (1-1/n) \bar{\phi}(\valvec{t-1}). \]
Taking expectations on both sides and applying this recursively implies
$\E{\bar{\phi}(\Valvec{t})} \geq (1-1/n)^{t} \bar{\phi}(\valvec{0})$.
Thus choosing    $t=o\left( n \ln\left(\frac{\valvec{0} }{\sigma^2 n}\right) \right)$ yields $\E{\bar{\phi}(\Valvec{t})} = \omega(\sigma^2 n)$~.
  
\textit{Proof of (ii).} Fix an arbitrary potential at round $t'$ and consider the next $n$ iterations. W.l.o.g. there has to be a constant fraction of the nodes with a value of greater or equal to the running average at time $t'$; otherwise there has to be such a fraction of nodes that have a value strictly smaller than the running average, in which case the proof is symmetric.
 Let $S$ be the set of these nodes. Order the nodes of $S$ according to their value in decreasing order (ties broken arbitrarily). Assign the first $\floor{|S|/2}$ to $S_1$ and the remaining $\ceil{|S|/2}$ nodes to $S_2$.
There will be w.h.p. a set $S_i'$, $i\in\{1,2\}$ of linear size in $n$ of nodes of $S_i$ that are chosen exactly once to exchange with another node of $S_i$ during the last $n$ steps and these nodes were not part of any other exchanges during the last $n$ steps.
Now consider the exchange of two nodes that belong to $S_1'$: the node with the initially lower value, will after averaging have with constant probability 
a value that is by $\Omega(\sigma)$ larger than before.
Similarly,  consider the exchange of two nodes that belong to $S_2'$: the node with the initially higher value, will after averaging have with constant probability 
a value that is by $\Omega(\sigma)$ smaller than before.

  Hence, by definition of the running average, irrespective of value of the running average at time $t$, the potential is of size 
$ \Omega(n\cdot \sigma^2)$ (due to the nodes of $S_1'$ or due to the nodes of $S_2'$).

\arxiv{add a step here at the end if possible}
\arxiv{can we also get a probabilistic statement with $\delta$?}

\ynote{Something went wrong in the above proof..just stopped}

\end{proof}

\section{Synchronous Model}\label{sec:synch}
In this section we consider the synchronous model and show that it is up to scaling of a factor of $n/2$ almost the same.
In order to avoid confusion, we introduce for every variable $V$ in the sequential model the synchronous/parallel counterpart $^\parallel V$  to emphasize the different model and the slightly different notation.
The following two lemmas are the counterparts of \autoref{lem:exactdistribution} and \autoref{lem:rewritingpotential}. These two
lemmas encapsulate the essential difference between both models.
\begin{lemma}[Synchronous Setting]\label{lem:exactdistributionsynch}
There exists random variables $N^*$,  $N'$, and $^\parallel\Delta^{(t+1)}$ s.t.
\[ ^\parallel\bar{\phi}(\Valvec{t+1})- ^\parallel\negmedspace\bar{\phi}(\valvec{t})  =  - ^\parallel\Delta^{(t+1)} \cdot {}^\parallel\bar{\phi}(\valvec{t})+  \sum_{i=1}^n\frac{{N'_i}}{4}+  \sum_{i=1}^n N^*_i \left(x_i - \avg{t} \right), \]
where
\[ \E{^\parallel\Delta^{(t+1)}~|~\mathcal{F}_t} = \frac12 \]
In particular, in the Gaussian noise model, we have
$N^*\sim \mathcal{N}(0,2\sigma^2)$ and $N'\sim \Gamma(1,2\sigma^2)$, where $\Gamma(\cdot,\cdot)$ denotes the gamma distribution.
\end{lemma}
This follows almost immediately from the sequential counter-part. In order to calculate the expectation, we can simply use linearity of expectation and multiply the sequential bound by a factor $n/2$

\begin{proposition}[Synchronous Setting]\label{lem:rewritingpotentialsynch}
Consider the interval $\mathcal{I} = (t_0,t_1]$.
Let  $^\parallel S'= \sum_{\tau \in \mathcal{I}} \sum_{i=1}^n N'^{(\tau)}_i/4$, \\ let $^\parallel S^*= \sum_{ \tau \in \mathcal{I}} N^{*{(\tau)}}  \left(\val{i}{\tau}- \avg{\tau} \right) $
and let $^\parallel S^-= \sum_{ \tau \in \mathcal{I}} ^\parallel \Delta^{(\tau)}  $.
We have that

\begin{align*}
^\parallel\bar{\phi}(\Valvec{t_1}) &\leq   \left(1- \frac{ ^\parallel S^- }{t}\right)^t {^\parallel}\bar{\phi}(\Valvec{t_0})  + ^\parallel\negmedspace S' + ^\parallel\negmedspace S^*.
\end{align*}  
\end{proposition}

The rest of the analysis is a straight-forward adaption of the  sequential setting with time being scaled by a factor of $n/2$.

\section{The Influence of Rounding}\label{sec:rounding}
The rounding can be implemented as follows assuming that the noise $N \sim \aleph$ takes only integer variables.
After a node $i$ receives the value from node $j$, the node averages it as before and then rounds up or down with equal probability. In symbols, 

\[ \Val{i}{t+1}  =  \begin{cases}
 \ceil{\frac{\val{i}{t} + \val{j}{t} + N}{2} } & \text{ w.p. $\frac12$}\\
 \\
 \floor{\frac{\val{i}{t} + \val{j}{t} + N}{2} } & \text{ otherwise}
 \end{cases},
 \]
where $N\sim \aleph$ is the integer valued channel noise. Equivalently, we can
write
\[ \Val{i}{t+1}  =  \begin{cases}
 {\frac{\val{i}{t} + \val{j}{t} + N+R}{2}} & \text{ w.p. $\frac12$}\\
 \\
 {\frac{\val{i}{t} + \val{j}{t} + N+R}{2}} & \text{ otherwise}
 \end{cases},
 \]
where $R$ is the random variable satisfying 
\[ R= \begin{cases}
\phantom{+}0 & \text{ if $\val{i}{t} + \val{j}{t} + N$ is even}\\
\phantom{+}1 & \text{ w.p. $\frac12$ provided that $\val{i}{t} + \val{i}{t} + N$ is odd}\\
-1 & \text{ otherwise}
 \end{cases}.
 \] 
 Regardless of the current state, it holds that $\E{R~|~\mathcal{F}_t}=0$ and $\Var{R~|~\mathcal{F}_t}=\E{R^2~|~\mathcal{F}_t} \leq 1$. Thus we obtain that $\E{(N+R)^2} = \E{N^2+2NR +  R^2} \leq  \E{N^2} + 1$. 
If we substitute $N$ with $N+R$ in all proofs we obtain essentially the same results; the variance in the  statements only increases by $1$.

\end{document}